\documentclass{CSML}
\pdfoutput=1
\usepackage{lastpage}
\newcommand{\dOi}{14(1:4)2018}
\lmcsheading{}{1--\pageref{LastPage}}{}{}%
{Mar.~16,~2017}{Jan.~16, 2018}{}


\usepackage{hyperref}
\hypersetup{hidelinks}
\usepackage{negotiations}
\usepackage{amssymb}
\usepackage{eucal}
\usepackage{graphicx}
\usepackage{amsmath}
\usepackage{xcolor}
\usepackage{stmaryrd}
\usepackage{rotating}
\usepackage{xspace}
\usepackage{eepic}
\usepackage{epsfig}
\usepackage{latexsym} 
\usepackage{exscale}
\usepackage{verbatim}

\AtBeginDocument{}

\newcommand{\es}{\emptyset}
\newcommand{\incl}{\subseteq}

\newcommand{\fleq}{\preccurlyeq}

\newcommand{\all}{\forall}

\newcommand{\fin}{\operatorname{\mathrm{fin}}}

\newcommand{\sq}[1]{[#1]}
\newcommand{\di}[1]{\langle #1 \rangle}
\newcommand{\sqq}[1]{\sq{\cdot }}
\newcommand{\ddi}[1]{\di{\cdot }}
\newcommand{\set}[1]{\{#1\}}

\def\sub[#1/#2]{[#1/#2]}

\newcommand{\act}[1]{\xrightarrow{#1}}       

\renewcommand{\d}{\delta}

\newcommand{\n}{\nu}

\newcommand{\s}{\sigma}

\newcommand{\x}{\xi}

\renewcommand{\S}{\Sigma}

\newcommand{\Bb}{\mathcal{B}}
\newcommand{\Cc}{\mathcal{C}}

\newcommand{\Ff}{\mathcal{F}}

\newcommand{\Nn}{\mathcal{N}}
\newcommand{\Oo}{\mathcal{O}}
\newcommand{\Pp}{\mathcal{P}}

\newcommand{\NLOGSPACE}{\text{\sc Nlogspace}}
\newcommand{\PSPACE}{\text{\sc Pspace}}
\newcommand{\PTIME}{\text{\sc Ptime}}
\newcommand{\NP}{\text{\sc NP}}
\newcommand{\coNP}{\text{\sc coNP}}

\newcommand{\struct}[1]{\langle #1 \rangle}

\usepackage[disable,textsize=small]{todonotes}

\newcommand{\igw}[1]{\todo[color=green!50]{\small #1}}
\newcommand{\igwin}[1]{\todo[color=green!30,inline]{#1}}
\newcommand{\anca}[2][]{\todo[color=blue!30,#1]{anca: #2}}

\newcommand{\SKIP}[1]{#1}

\newcommand{\Nnd}{\Nn_D}
\newcommand{\GNB}{G(\Nn,B)}

\newcommand{\smax}{\s_{max}}
\newcommand{\leqN}{\fleq_{\Nn}}
\newcommand{\leN}{\prec_{\Nn}}

\newcommand{\leNd}{\prec_{\Nn_D}}
\newcommand{\init}{\mathit{init}}
\renewcommand{\fin}{\mathit{fin}}
\newcommand{\Cinit}{C_{\init}}
\newcommand{\Cfin}{C_{\fin}}

\newcommand{\N}{\mathcal N}

\newcommand{\dom}{\mathit{dom}}
\newcommand{\out}{\mathit{out}}
\newcommand{\Proc}{\mathit{Proc}}

\newcommand{\rd}{\mathit{read}}
\newcommand{\wrt}{\mathit{write}}
\renewcommand{\all}{\mathit{alloc}}
\newcommand{\deall}{\mathit{dealloc}}

\newcommand{\para}{\parallel}

\newcommand{\tr}{\mathit{tr}}

\begin{document}

\title[Soundness in negotiations]{Soundness in negotiations$^*$}
\titlecomment{{\lsuper*}A preliminary version of this paper appeared
  in \cite{ekmw16}. \\
This work was carried out while Anca Muscholl 
was on leave at the Institute of Advanced Studies of the 
Technical University of Munich, supported by a Hans-Fischer Fellowship, 
and Denis Kuperberg had a postdoctoral position at the same institute.\\
The work was also supported by the 
DFG Project ``Negotiations: A Model for Tractable Concurrency''.}

\author[J.~Esparza]{Javier Esparza\rsuper{a}}
\address{\lsuper{a}Technical University of Munich}
\email{esparza@in.tum.de}
\author[D.~Kuperberg]{Denis Kuperberg\rsuper{b}}
\address{\lsuper{b}\'Ecole Normale Sup\'erieure de Lyon}
\email{denis.kuperberg@ens-lyon.fr}
\author[A.~Muscholl]{Anca Muscholl\rsuper{c}}
\address{\lsuper{c}LaBRI, University of Bordeaux}
\email{\{anca,igw\}@labri.fr}
\author[I.~Walukiewicz]{Igor Walukiewicz\rsuper{c}}

\begin{abstract}
Negotiations are a formalism for describing multiparty distributed
cooperation. Alternatively, they can be seen as a model of
concurrency with synchronized choice as communication
primitive. 

Well-designed negotiations must be sound, meaning that,
whatever its current state, the negotiation can still be
completed.  In earlier work, Esparza and Desel have shown that deciding
soundness of a negotiation is \PSPACE-complete, and in \PTIME\ if the 
negotiation is deterministic. They have also extended their 
polynomial soundness algorithm to an intermediate class of acyclic, 
non-deterministic negotiations. However, they did not analyze the 
runtime of the extended algorithm, and also left open the complexity of the soundness problem for the intermediate class.

In the first part of this paper we revisit the soundness problem for
deterministic negotiations, and show that it is \NLOGSPACE-complete,
improving on the earlier algorithm, which requires linear space.  

In the second part we answer the question left
open by Esparza and Desel. We prove that the soundness problem can be
solved in polynomial time for acyclic, weakly non-deterministic
negotiations, a more general class than the one considered by them. 

In the third and final part, we show that the techniques developed in
the first two parts of the paper can be applied to analysis problems 
other than soundness, including the problem of detecting race conditions, and 
several classical static analysis problems. More specifically, we show that, 
while these problems are intractable for arbitrary acyclic 
deterministic negotiations, they become tractable in the sound case. 
So soundness is not only a desirable behavioral property in itself, but also helps to analyze other properties. 
\end{abstract}

\maketitle

\section{Introduction}
\label{sec:intro}

A multiparty atomic negotiation is an event in which several
processes (agents) synchronize in order to select one out of a number of
possible results. In \cite{negI} Esparza and Desel introduced
\emph{negotiation diagrams}, or just \emph{negotiations}, a model of concurrency with multiparty atomic negotiation as interaction primitive. 
A  negotiation diagram describes a workflow of ``atomic''
negotiations. After an atomic negotiation concludes with the selection
of a result, the workflow determines the set of atomic negotiations
each agent is ready to participate in next. 

Negotiation diagrams are closely related to workflow Petri nets, a very successful formalism for the description of business processes, 
and a back-end for graphical notations like BPMN (Business Process Modeling Notation), 
EPC (Event-driven Process Chain), or UML Activity Diagrams (see e.g.\ \cite{aalst,van2004workflow}).  In a nutshell, negotiation diagrams are workflow Petri nets that can be decomposed into communicating 
sequential Petri nets, a feature that makes them more 
amenable to theoretical study, while the translation into workflow nets 
(described in \cite{DBLP:journals/topnoc/DeselE16}) allows to transfer results and algorithms to business process applications. 

The most prominent analysis problem for the negotiation model is
checking \emph{soundness}, a notion originally introduced for workflow
Petri nets. Loosely speaking, a negotiation is sound if from every reachable configuration 
there is an execution leading to proper termination of the negotiation. 
In \cite{negI} it is shown that the
soundness problem is \PSPACE-complete for non-deterministic
negotiations and \coNP-complete for acyclic non-deterministic
negotiations.  For this reason, and in search
of a tractable class, \cite{negI} introduces the class of
\emph{deterministic negotiations}. In deterministic negotiations all
agents are deterministic, meaning that they are never ready to engage
in more than one atomic negotiation per result (similarly to 
a deterministic automaton, that for each action can 
move to at most one state). In \cite{negI} the soundness problem is investigated
for acyclic negotiations. The main results  are a
polynomial time algorithm for checking soundness of deterministic negotiations, 
and an extension of the algorithm to the more expressive class of weakly
deterministic negotiations. However, whether the extended
algorithm is polynomial or not was left open.  In \cite{negII} the polynomial result 
for acyclic deterministic negotiations is extended to the cyclic case.

In this paper we continue the line of research initiated in \cite{negI,negII},
and present three contributions. 
  
In the first contribution we revisit the soundness problem for
deterministic negotiations.  It should be noted that the notion of soundness in \cite{negI} has one
  more requirement (which makes the soundness problem for acyclic
  negotiations \coNP-hard and in DP). We show here that for
  deterministic, possibly cyclic, negotiations this second
  requirement is unnecessary, modulo a weak assumption saying that
  every atomic negotiation is reachable from the initial negotiation
  by a local path in the graph on the negotiation.
\igwin{added explanation}
We then identify \emph{anti-patterns}, i.e., structures of the graph of a negotiation, 
and show that a deterministic negotiation is unsound if{}f it exhibits at least one of them. 
As an easy consequence of this theorem, we obtain an \NLOGSPACE\ algorithm for checking 
soundness, whereas the algorithm of \cite{negII} requires linear space. Since soundness of deterministic negotiations is easily shown to be \NLOGSPACE-hard, our result settles the 
complexity of the soundness problem for deterministic negotiations.

In the second contribution we answer the question left
open in \cite{negI}. We prove that the soundness problem can be
solved in polynomial time for acyclic, weakly non-deterministic
negotiations, a class that is more general than the one considered in
\cite{negI}\footnote{The weakly
deterministic negotiations of \cite{negI} are called \emph{very
 weakly non-deterministic negotiations} in this paper. As the name indicates,
every negotiation in this class is also weakly
non-deterministic. \anca[inline]{add the assumption?}}. The result
is based on a game-theoretic solution to the \emph{omitting problem}, an 
analysis problem of independent interest.
Further, we show that if we leave out one of the two
assumptions, acyclicity or weak non-determinism, then the problem
becomes \coNP-complete\footnote{We show that \coNP-hardness holds even for a very mild 
relaxation of acyclicity.}. These results set a limit to the class of 
negotiations with a polynomial soundness problems, but also admit a positive interpretation. 
Indeed, the soundness problem for arbitrary negotiations is \PSPACE-complete \cite{negI}, 
and so of higher complexity (under the usual assumption 
$\PTIME \subset \NP \cup \coNP \subset \PSPACE$).
  
In the third and final contribution, we show that the techniques developed in
the first two parts of the paper, namely anti-patterns and our game-theoretic
solution to the omitting problem, can be applied to analysis problems other than 
soundness. More specifically, we show that, while these problems are intractable 
for arbitrary deterministic negotiations, they become tractable in the sound case. 
So soundness is not only a desirable behavioral property in itself, but also helps to analyze other properties. The first problem we consider is the 
existence of \emph{races}, i.e., executions in which two
given atomic negotiations are concurrently enabled. We show that for acyclic deterministic negotiations the problem is in \NLOGSPACE. Then we analyze several classical 
program analysis problems for negotiations that manipulate data, 
for example whether every value written into a variable is guaranteed
to be read.
Such problems have been studied for 
workflow nets in \cite{DBLP:conf/caise/TrckaAS09,DBLP:journals/is/SidorovaST11},
and exponential algorithms have been proposed. We show that for acyclic sound deterministic 
negotiations the problems can be solved in polynomial time. 

\subsubsection*{Related formalisms and related work.} 
The connection between negotiations and Petri nets is
studied in detail in \cite{DBLP:journals/topnoc/DeselE16}.  The
connection is particularly close between deterministic negotiations
and free-choice workflow nets. The complexity of the soundness problem
for workflow nets has been studied in several papers \cite{DBLP:conf/bpm/Aalst00,DBLP:journals/fac/AalstHHSVVW11,DBLP:journals/tsc/Liu14a,DBLP:journals/tsmc/TipleaBC15},
and in particular in \cite{DBLP:conf/tacas/FavreVM16} soundness 
of free-choice workflow nets is also characterized in terms of
anti-patterns, which can be used to explain why a given 
workflow net is unsound. In the conclusions we discuss the consequences of our results for 
the analysis of soundness in workflow Petri nets in detail.

As a process-based concurrent model, negotiations can be compared
with another well-studied model for distributed computation,
namely Zielonka automata~\cite{zie87,DR95,mus15}.
Such an automaton is a parallel composition of finite transition
systems with synchronization on common actions.  
The important point is that a synchronization in Zielonka automata involves exchange of information
between states of agents: the result of the synchronization depends on
the states of all the components taking part in it. 
Zielonka automata have the same expressive power as arbitrary,
possibly nondeterministic negotiations. Deterministic negotiations 
correspond to a subclass that does not seem to have been studied yet, and
for which verification becomes considerably easier. For example, the question
whether some local state occurs in some execution is \PSPACE-complete
for ``sound'' Zielonka automata, while it can be answered in polynomial
time for sound deterministic negotiations.

A somewhat similar graphical formalism are message sequence
charts/graphs, used to describe asynchronous communication. Questions like
non-emptiness of intersection are in general undecidable for this model, even
assuming that communication buffers are bounded. Subclasses of message
sequence graphs with decidable model-checking problem were proposed, but 
the complexity is \PSPACE-complete~\cite{GKM06}.

\subsubsection*{Overview.} Section~\ref{sec:def} introduces definitions and notations.
Section~\ref{sec:det} revisits the soundness problem for deterministic
negotiations and is new compared to the conference version~\cite{ekmw16}. Section \ref{sec:beyond-det-tract} shows that soundness of acyclic weakly 
non-deterministic negotiations can be decided in polynomial time; the first part
of the section solves the omitting problem, and the second part applies the solution to
the soundness problem. Section \ref{sec:beyond-det-intract} proves that dropping 
acyclicity or weak non-determinism makes the soundness problem intractable. 
Section \ref{sec:beyond-sound} gives polynomial algorithms for the race problem
and the static analysis problems of sound deterministic negotiations.
Section \ref{sec:conc} presents our conclusions.


\section{Negotiations}\label{sec:def}

A~\emph{negotiation} $\Nn$ is a tuple
$\struct{\Proc,N,\dom,R,\d}$, where 
$\Proc$ is a finite set of \emph{processes} (or agents) that can participate in negotiations, and $N$ is a finite set of \emph{nodes} (or \emph{atomic negotiations}) where the processes can synchronize.
The function $\dom:N\to \Pp(\Proc)$ associates to every atomic
negotiation $n \in N$ the (non-empty) set $\dom(n)$ of processes participating in it
(\emph{domain} of $n$). 
Negotiations come equipped with two distinguished initial and final atomic
negotiations $n_\init$ and $n_\fin$ in which {\em all} processes in
$\Proc$ participate. Nodes are denoted as $m, n, \ldots$, and processes as 
$p, q,\ldots$, possibly with indices. 

The set of possible results of atomic negotiations is denoted $R$,
and we use $a,b,\ldots$ to range over its elements. Every atomic negotiation $n\in N$ has its set of possible results $\out(n)\subseteq R$. 
We assume that every atomic negotiation (except possibly for $n_\fin$) has at least one result.
(This is a slight change with respect to the definitions of \cite{negI,negII},
due to the fact that the final result is not relevant for the
soundness question.) 
The control flow
in a negotiation is determined by a partial transition function
$\d:N\times R\times \Proc\act{\cdot} \Pp(N)$, telling that after the
result $a \in\out(n)$ of an atomic negotiation $n$, process $p\in\dom(n)$ is
ready to participate in any of the negotiations from the set
$\d(n,a,p)$.  
So for every $n'\in\d(n,a,p)$ we have $p\in
\dom(n')\cap\dom(n)$. For every $n$, $a\in \out(n)$ and
$p\in\dom(n)$ the result $\d(n,a,p)$ has to be defined and non-empty. 
So all
processes involved in an atomic negotiation should be ready for all
its possible results. Observe that atomic negotiations may
have one single participant process, and/or have one single result.

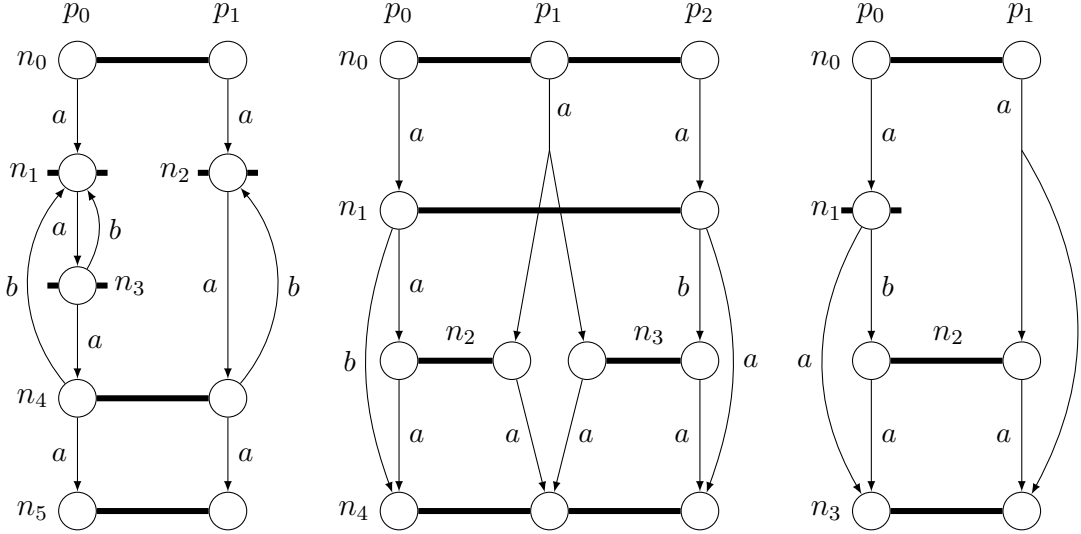
\begin{figure}[ht]
\centerline{
\scalebox{1.0}{\begin{tikzpicture}
\nego[ports=2,id=n0,spacing=2]{0,0}
\node[left = 0cm of n0_P0, font=\large] {$n_0$};
\nego[ports=1,id=n1, spacing=2]{0,-1.5}
\node[left = 0.1cm of n1_P0, font=\large] {$n_1$};
\nego[ports=1,id=n2]{2,-1.5}
\node[left = 0.1cm of n2_P0, font=\large] {$n_2$};
\nego[ports=1,id=n3]{0,-3}
\node[right = 0.1cm of n3_P0, font=\large] {$n_3$};
\nego[ports=2,id=n4, spacing=2]{0,-4.5}
\node[left = 0cm of n4_P0, font=\large] {$n_4$};
\nego[ports=2,id=n5, spacing=2]{0,-6}
\node[left = 0cm of n5_P0, font=\large] {$n_5$};

\node[above = 0.1cm of n0_P0, font=\large] {$p_0$};
\node[above = 0.1cm of n0_P1, font=\large] {$p_1$};

\pgfsetarrowsend{latex}
\draw (n0_P0) edge node[left=0] {$a$} (n1_P0);
\draw (n0_P1) edge node[right=0] {$a$} (n2_P0);
\draw (n1_P0) edge node[left=0] {$a$} (n3_P0);
\draw (n2_P0) edge node[left=0] {$a$} (n4_P1);
\draw (n3_P0) edge node[right=0] {$a$} (n4_P0);
\draw (n3_P0) edge[bend right=30] node[right=0] {$b$} (n1_P0);
\draw (n4_P0) edge[bend left=40] node[left=0] {$b$} (n1_P0);
\draw (n4_P1) edge[bend right=40] node[right=0] {$b$} (n2_P0);
\draw (n4_P0) edge node[left=0] {$a$} (n5_P0);
\draw (n4_P1) edge node[right=0] {$a$} (n5_P1);

\end{tikzpicture}}
\scalebox{1.0}{\begin{tikzpicture}
\nego[ports=3,id=n0,spacing=2]{0,0}
\node[left = 0cm of n0_P0, font=\large] {$n_0$};
\nego[ports=2,id=n05,spacing=4]{0,-2}
\node[left = 0cm of n05_P0, font=\large] {$n_1$};
\nego[ports=2,id=n1, spacing=1.5]{0,-4}
\node[above right = -0.1cm and 0.3cm of n1_P0, font=\large] {$n_2$};
\nego[ports=2,id=n2,spacing=1.5]{2.5,-4}
\node[above right = -0.1cm and 0.3cm of n2_P0, font=\large] {$n_3$};
\nego[ports=3,id=n3, spacing=2]{0,-6.0}
\node[left = 0cm of n3_P0, font=\large] {$n_4$};

\node[above = 0.1cm of n0_P0, font=\large] {$p_0$};
\node[above = 0.1cm of n0_P1, font=\large] {$p_1$};
\node[above = 0.1cm of n0_P2, font=\large] {$p_2$};

\pgfsetarrowsend{latex}
\draw (n0_P0) edge node[right=0] {$a$} (n05_P0);
\draw (n0_P2) edge node[left=0] {$a$} (n05_P1);
\draw (n05_P0) edge node[right=0] {$a$} (n1_P0);
\draw (n05_P1) edge node[left=0] {$b$} (n2_P1);
\draw (n1_P0) edge node[right=0] {$a$} (n3_P0);
\draw (n1_P1) edge node[left=0] {$a$} (n3_P1);
\draw (n2_P0) edge node[right=0] {$a$} (n3_P1);
\draw (n2_P1) edge node[left=0] {$a$} (n3_P2);
\draw (n05_P0) edge[bend right = 20] node[left=0] {$b$} (n3_P0);
\draw (n05_P1) edge[bend left = 20] node[right=0] {$a$} (n3_P2);
\draw (n0_P1) to +(0,-1.2) -- (n1_P1);
\labelEdge[start=n0_P1,end={2,-1.5},dist=0.75,shift={0.2,0}]{$a$}
\draw (n0_P1) to +(0,-1.2) -- (n2_P0);

\end{tikzpicture}}
\scalebox{1.0}{\begin{tikzpicture}
\nego[ports=2,id=n0,spacing=2]{0,0}
\node[left = 0cm of n0_P0, font=\large] {$n_0$};
\nego[ports=1,id=n1, spacing=2]{0,-2}
\node[left = 0cm of n1_P0, font=\large] {$n_1$};
\nego[ports=2,id=n2,spacing=2]{0,-4}
\node[above right = -0.1cm and 0.5cm of n2_P0, font=\large] {$n_2$};
\nego[ports=2,id=nf,spacing=2]{0,-6}
\node[left = 0cm of nf_P0, font=\large] {$n_3$};

\node[above = 0.1cm of n0_P0, font=\large] {$p_0$};
\node[above = 0.1cm of n0_P1, font=\large] {$p_1$};

\pgfsetarrowsend{latex}
\draw (n0_P0) edge node[right=0] {$a$} (n1_P0);
\draw (n0_P1) to +(0,-1.2) -- (n2_P1);
\labelEdge[start=n0_P1,end={0,-1.5},dist=0.7,shift={0.7,0}]{$a$}
\path (n0_P1) -- +(0,-1.2) edge[bend left=30] (nf_P1);

\draw (n1_P0) edge node[right=0] {$b$} (n2_P0);
\draw (n1_P0) edge[bend right=30] node[left=0] {$a$} (nf_P0);
\draw (n2_P0) edge node[right=0] {$a$} (nf_P0);
\draw (n2_P1) edge node[left=0] {$a$} (nf_P1);
\end{tikzpicture}}}
\caption{Three negotiations. In all of them, $n_0$ is the initial node and the final node is the bottom one.}
\label{fig:example1}
\end{figure}

Negotiations admit a graphical representation, cf.\
Figure~\ref{fig:example1}. A node (atomic negotiation) $n$ is
represented as a black bar with a white circles, called {\em ports},
one for each process in $\dom(n)$. So, for example, the negotiation on
the left of Figure~\ref{fig:example1} has six nodes
$n_0, \ldots, n_5$. Nodes $n_0$, $n_4$, and $n_5$ have two ports each,
while nodes $n_1$, $n_2$, and $n_3$ only have one port. An entry
$\delta(n, a, p) = \{n_1, \ldots, n_k\}$ is represented by a
hyper-arc, labeled by $a$, that connects the port of process $p$ in
$n$ with the ports of process $p$ in $n_1, \ldots, n_k$. In
particular, for the negotiation on the left of
Figure~\ref{fig:example1} we have $\Proc = \{ p_0, p_1 \}$,
$N = \{n_0, \ldots, n_5\}$, $R = \{a, b\}$, and, for example
$\dom(n_1)=\{p_0\}$, $\dom(n_4) = \{p_0, p_1\}$,
$\d(n_4,b,p_0)= \{n_1\}$ and $\d(n_4,b,p_1)= \{n_2\}$.  This
negotiation does not contain any proper hyper-arcs, but the second one
does; there we have $\d(n_0,a,p_1) = \{n_2, n_3\}$. In general two nodes
may be connected by several arcs, carrying different process/result
labels. For instance, node $n_4$ in the negotiation on the left of
Figure~\ref{fig:example1} has two arcs to $n_5$.

\subsection*{Configurations} A \emph{configuration} of a negotiation is a function 
$C \colon \Proc\to \Pp(N)$ mapping each process $p$ to the (non-empty) set of atomic 
negotiations in which
$p$ is ready to engage. The \emph{initial and final
configurations} $\Cinit$, $\Cfin$ are given by $\Cinit(p)=\set{n_\init}$
and $\Cfin(p)=\set{n_\fin}$ for all $p\in \Proc$. 
An atomic negotiation $n$ is \emph{enabled} in a 
configuration $C$ if $n\in C(p)$ for every $p\in \dom(n)$, that is, if all 
processes that participate in $n$ are ready to proceed with it. A configuration 
is a \emph{deadlock} if no atomic negotiation is enabled in it. If an atomic
negotiation $n$ is enabled in $C$, and $a$ is a result
of $n$, then we say that \emph{$(n,a)$ can be executed}, and its 
execution produces a new configuration $C'$ given by $C'(p)=\d(n,a,p)$ 
for $p\in\dom(n)$ and $C'(p)=C(p)$ for $p\not \in \dom(n)$.  We denote this by
$C\act{(n,a)}C'$. For example, in the negotiation on the right of Figure~\ref{fig:example1} we have 
$C\act{(n_0,a)}C'$ for $C(p_0) = \{n_0\}= C(p_1)$ and
$C'(p_0) = \{n_1\}, C'(p_1)=\{n_2, n_3\}$. Observe that in all three negotiations,
the final configuration cannot be executed, since the final node has
no result. However, by definition the final node is enabled in the
final configuration, and so the final configuration is not a deadlock.\footnote{While an enabled atomic negotiation that cannot be executed is a bit artificial, we adopt it because 
it allows to separate deadlock and termination.}

\subsection*{Runs} A~\emph{run} of a negotiation $\Nn$ from a configuration $C_1$ is a finite or infinite sequence
$w=(n_1,a_1)(n_2,a_2)\dots$ such that there are configurations
$C_2,C_3,\dots$ with 
\begin{equation*}
  C_1\act{(n_1,a_1)} C_2\act{(n_2,a_2)} C_3\cdots
\end{equation*}
We denote this by $C_1\act{w}$, or $C_1\act{w}C_k$ if the sequence
is finite and finishes with $C_k$.
In the latter case we say that $C_k$ is \emph{reachable from $C_1$ on $w$}.
We simply call it \emph{reachable} if $w$ is irrelevant, and write
$C_1\act{*}C_k$. Consider for example the third negotiation of Figure \ref{fig:example1}.
If we represent a configuration $C$ by the tuple $(C(p_0), C(p_1))$ then 
$$ (\{n_0\}, \{n_0\}) \act{(n_0, a)} (\{n_1\}, \{n_2, n_3\}) \act{(n_1, b)} (\{n_2\}, \{n_2, n_3\})
\act{(n_2, a)} (\{n_3\}, \{n_3\})$$
\noindent is a run.

A run is called \emph{initial} if it starts in $\Cinit$. An initial run is
\emph{successful} if it starts in $\Cinit$ and ends in
$\Cfin$. In the three negotiations of Figure~\ref{fig:example1} we 
have $n_\init = n_0$, and $n_\fin = n_5, n_4, n_3$, respectively. The run shown above is both initial
and successful. 

\subsection*{Acyclicity} The \emph{graph of a negotiation}\label{def:graph} has $N$, the set of atomic negotiations, as
the set of vertices;  the edges are $n\act{p,a} n'$ if  $n'\in \d(n,a,p)$. Observe
that $p\in \dom(n)\cap\dom(n')$.

A negotiation is \emph{acyclic} if its graph is so. 
For an acyclic negotiation $\Nn$ we fix a linear order $\leqN$ on its nodes
that is a topological order on the graph of $\Nn$. 
This means that if there is an edge from $m$ to $n$ in the graph of
$\Nn$ then $m\leqN n$. The last two negotiations of Figure~\ref{fig:example1} are
acyclic, while the first one is not. For the negotiation in the middle of the figure 
there are two options for the topological order, corresponding to fixing 
$n_2 \leqN n_3$ or $n_3 \leqN n_2$.

\subsection*{Paths} 
Fix a negotiation $\Nn$.  A~\emph{local path} of $\Nn$ is a path
$n_0\act{p_0,a_0}n_1\act{p_1,a_1}\dots\act{p_{k-1},a_{k-1}} n_k$ in
the graph of $\Nn$. A local path is 
\begin{itemize}
\item a \emph{circuit} if $n_0 = n_k$ and $k\geq 1$;
\item a \emph{$p$-path} if $p_0 = \cdots = p_{k-1} = p$;
\item \emph{realizable} from a configuration $C$ if there is a run
$$C \act{(n_0,a_0)} C_0' \act{w_1} C_1 \act{(n_1, a_1)} C_1' \cdots C_{k-1}\act{(n_{k-1},a_{k-1})} C_{k-1}' \act{w_k} C_k$$
\noindent such that $p_i \notin \dom(w_{i+1})$ for all
$i=0,\ldots,k-1$. (Here $\dom(v)$ denotes the set of processes involved
in $v$, that is, $\dom(v)=\bigcup\set{\dom(n) : \text{$(n,a)$ appears in $v$ for some $a \in \out(n)$}}$.)
We say that the run \emph{realizes} the path from $C$.
\end{itemize}

\begin{exa}
For example, $n_0 \act{p_1,a} n_2 \act{p_1,a} n_4 \act{p_0, a} n_5$
is a local path of the first negotiation of Figure \ref{fig:example1}.
The path is realized from the initial configuration by the run 
$$\begin{array}{ll}
(\{n_0\}, \{n_0\}) & \act{(n_0, a)} (\{n_1\}, \{n_2\}) \act{(n_1, a)} (\{n_3\}, \{n_2\}) \\[0.2cm]
& \act{(n_3, a)} (\{n_4\}, \{n_2\}) \act{(n_2, a)} (\{n_4\}, \{n_4\}) \act{(n_4, a)} (\{n_5\}, \{n_5\})
\end{array}$$
\noindent Indeed, we can take $w = (n_0, a) \, w_1 \, (n_2, a) \, w_2 \, (n_4, a)$, with 
$w_1 = (n_1, a) \, (n_3, a)$ and $w_2 = \epsilon$. 
\end{exa}

\subsection*{Soundness} A negotiation $\Nn$ is \emph{sound} if every initial run 
can be completed to a successful run. If a negotiation has no infinite runs
(for example, this is the case if the negotiation is acyclic), 
then it is sound if{}f it has no reachable deadlock configuration. The three negotiations of 
Figure~\ref{fig:example1} are sound. If in the negotiation on the right
we change $\delta(n_0, a, p_1)$ from $\{n_2, n_3\}$ to $\{n_3\}$, then 
the negotiation is no longer sound. Indeed, after the change the negotiation has the run
$$ (\{n_0\}, \{n_0\}) \act{(n_0, a)} (\{n_1\}, \{n_3\}) \act{(n_1, b)} (\{n_2\}, \{n_3\})$$
which leads to a deadlock. 

\begin{rem}
\label{rem:sounddef}
Our definition of soundness is slightly different from the one used in \cite{negI}.
The definition of \cite{negI}, which follows the definition of soundness for workflow 
Petri nets introduced in \cite{aalst}, requires an additional property: for every atomic 
negotiation $n$ there is an initial run that enables $n$. We use the weaker definition because it 
leads to cleaner theoretical results, and because, as we shall see, the two definitions are essentially 
equivalent for deterministic negotiations (see Remark \ref{rem:soundequiv}).
\end{rem} 

For deterministic negotiations, being unsound is closely related to
deadlocks, as the observation below shows:\footnote{We could have
  use the weaker requirement that every accessible node has some local path, not
  necessarily $p$-path, to $n_\fin$. However, we use the precise
  statement of the lemma later.}

\begin{lem}\label{lem:unsound-deadlock}
  Let $\Nn$ be a deterministic negotiation such that for every node $n
  \in N$ that is accessible from $n_\init$, and every process $p
  \in\dom(n)$, there exists some $p$-path from $n$ to $n_\fin$. Then
  $\Nn$ is unsound iff some initial run leads to a deadlock
  configuration. 
\end{lem}

\begin{proof}
Since the right-to-left implication is obvious, we assume that $\Nn$
is unsound and we show that a reachable deadlock must exist.

For every atomic negotiation $n$ and $p \in \dom(n)$, let $d(n,p)$ be the length of the shortest 
$p$-path leading from $n$ to $n_\fin$ (or $\infty$ if the path does not exist), 
and for every reachable configuration $C$, define 
$d(C) = \left( d(C(p_1),p_1), \ldots, d(C(p_n),p_n)\right)$. Since $\Nn$ is unsound, some
initial run $C_\init \act{*} C$ cannot be extended to a successful run. Choose the run 
so that $d(C)$ is minimal w.r.t.~the lexicographic order according to
some order on $\Proc$. We claim that $C$ is a deadlock. Assume 
the contrary, so $C$ enables some node $n$. Let $p$ be the smallest
process of $\dom(n)$. By assumption, there is some $p$-path from $n$
to $n_\fin$. Let $a \in\out(n)$ and $n'\in N$ be such that $n'=\delta(n,a,p)$ is the
first node on the shortest $p$-path from $n$ to $n_\fin$. In particular, we have $d(n',p) < d(n,p)$.
 Taking $C \act{(n,a)} C'$, we get that $d(C')$ is lexicographically smaller than $d(C)$, contradicting 
the minimality of $C$, and so the claim is proved. 
\end{proof}

\subsection*{Determinism} Process $p$ is \emph{deterministic} in a negotiation $\Nn$ if for
every $n\in N$ and every $a\in \out(n)$, the set of possible next negotiations,
$\d(n,a,p)$, is a singleton. 
A negotiation is \emph{deterministic} if every process $p\in\Proc$ is
deterministic. Graphically, a negotiation is deterministic if it does not
have any proper hyper-arc. The negotiation on the left of Figure~\ref{fig:example1} is deterministic.

A negotiation is \emph{weakly non-deterministic}\label{def:weak} if for every $n\in N$ 
at least one of the processes in $\dom(n)$ is deterministic.
A negotiation is \emph{very weakly non-deterministic}\label{def:very-weak}\footnote{This
class was called \emph{weakly deterministic} in~\cite{negI}.} if for every $n\in N$, 
$a\in \out(n)$, and $p\in\dom(n)$, there is a deterministic process $q$ such
that $q\in \dom(n')$ for all $n'\in \d(n,a,p)$.
As the names suggest, a very weakly non-deterministic negotiation is 
weakly non-deterministic, under the (very weak) assumption that every
node, but $n_\init$ is a target of some transition. 
\igwin{changed the assumption}
Indeed, if a node $n$ in a very weakly non-deterministic negotiation is
reachable from some other node then there is a deterministic process
in $\dom(n)$. The same is true if $n=n_\init$, since all 
processes are in the domain of the initial node. 
This shows the claim. 
The intuition behind very weak nondeterminism is that every
nondeterminism in the transition function should be resolved by a
deterministic process. 
\igw{added intuition}

The negotiation in the middle of Figure~\ref{fig:example1} is weakly non-deterministic. 
Indeed, the processes $p_0$ and $p_2$ are deterministic, and every node has $p_0$ or $p_2$
(or both) in its domain. However, it is not very weakly non-deterministic. To see this,
observe that $\d(n_0, a, p_1) = \{n_2, n_3\}$, but the intersection $\dom(n_2) \cap \dom(n_3) = \{p_1\}$
does not contain any deterministic process. On the contrary, the negotiation on the right of the figure 
is very weakly non-deterministic, because the deterministic process $p_0$ belongs to the domain of all nodes.

Weakly non-deterministic negotiations allow to model 
deterministic negotiations with global resources (see Section~\ref{sec:beyond-sound}). The resource 
(say, a piece of data) can be modeled as an additional process, which participates in the atomic 
negotiations that use the resource. For example, the negotiation 
in the middle of Figure~\ref{fig:example1} models a situation in which processes $p_0$ and $p_2$
negotiate in $n_1$ which of the two will have access to the resource modeled by $p_1$. If 
the result of $n_1$ is $a$, then $p_0$ has access to the resource at node $n_2$, and if it is $b$, then 
$p_2$ has access to it at node $n_3$.


\section{Soundness of deterministic negotiations}\label{sec:det}


We revisit the soundness problem for deterministic negotiations. 
We give the first \NLOGSPACE\ algorithm for the problem, 
in contrast with the polynomial algorithm of~\cite{negI}, 
which requires linear space. The algorithm is based on analysis of
the graph of a negotiation as defined on
page~\pageref{def:graph}. More precisely we show a novel
characterization of soundness in terms 
of \emph{anti-patterns} in  this graph.
The characterization allows not only to check soundness, but
also to diagnose why a given negotiation is unsound. 

The following lemma shows that every local path of a sound and deterministic negotiation
is realizable from some reachable configuration. 
\begin{lem}\label{l:local}
  Let $\pi$ be a local path of a sound deterministic negotiation $\Nn$, and let
  $n_0$ be the first node of $\pi$. Then $\pi$ is realizable from every reachable configuration
  that enables $n_0$.
\end{lem}
{
\begin{proof}
Let $\pi = n_0\act{p_0,a_0}n_1\act{p_1,a_1}\cdots\act{p_{k-1},a_{k-1}} n_k$, and 
let $C$ be a reachable configuration such that $C(p)=n_0$ for every $p \in\dom(n_0)$. 
  By induction on $i$ we show that there is a run $C\act{*} C_i$ 
  realizing
  $n_0\act{p_0,a_0}n_1\act{p_1,a_1}\cdots\act{p_{i-1},a_{i-1}} n_i$ and
  such that  $n_i$ is enabled in $C_i$. 
  
  For $i=0$, we simply take $C_i=C$.
For the induction step we assume the existence of $C_i$ in which $n_i$ is
  enabled. 
  Let $C'_{i+1}$ be the result of executing $(n_i,a_i)$ from $C_i$. 
  Observe that $C'_{i+1}(p_i)=n_{i+1}$ (recall that $\Nn$ is deterministic).
  Since $\Nn$ is sound, and $C'_{i+1}$ is reachable, there is a run
  from $C'_{i+1}$ to $C_\fin$. 
  We set then $C_{i+1}$ to be the first configuration on this run when
$n_{i+1}$ is enabled.
  \end{proof}
}

In particular, Lemma~\ref{l:local} states that there is an initial run containing the atomic
negotiation $m$ iff there is a local path from $n_\init$ to $m$.  
If $\dom(m) \cap \dom(n)\not=\es$ then the lemma also provides an easy
test for deciding the existence of a run containing both $m,n$: it
suffices to check the existence of a local path $n_\init \act{*} m
\act{*} n$, or with $m,n$ interchanged.

Our algorithm for checking soundness of deterministic negotiations
checks for certain patterns in the graph of the negotiation. 
Since negotiations exhibiting the patterns are unsound, we call them \emph{anti-patterns}.
In order to define them we need to introduce \emph{forks} and recall 
the notion of \emph{dominating node} of a local path introduced in \cite{negII}.

\begin{defi}\label{def:fork}
Let $\Nn =(\Proc,N,\dom, R,\delta)$ be a deterministic negotiation.
A tuple $(p_1,p_2,n_1,n_2) \in \Proc^2 \times N^2$ is a \emph{fork} of $\Nn$ if there exists a local path from $n_\init$ to a node $n \in N$ and a result $a \in \out(n)$ such that
\begin{itemize}
\item $p_i \in \dom(n) \cap \dom(n_i)$ for $i=1,2$;
\item for $i=1,2$ there exists a $p_i$-path $\pi_i$ leading from $\d(n,a,p_i)$ to $n_i$; and 
\item $\pi_1$ and $\pi_2$ are disjoint, i.e., no node appears in both.
\end{itemize} 
\end{defi}

\begin{defi}
A node $n$ of a local path $\pi$ \emph{dominates} $\pi$ if $\dom(m) \subseteq \dom(n)$ for every node $m$ of $\pi$.
\end{defi}

\begin{exa}
The tuple $(p_0, p_1, n_3, n_4)$ is a fork of the negotiation on the left of Figure~\ref{fig:example1}. 
We can choose $n=n_0$. The $p_0$-path is $n_1 \act{p_0, a} n_3$, and the $p_1$-path is 
$n_2 \act{p_1, a} n_4$. The tuple $(p_0,p_1,n_4,n_5)$
is not a fork since all $p_1$-paths from $n_0$ to $n_5$ go through
$n_4$, but the $p_0$-path and the $p_1$-path are required to be disjoint. If we change this negotiation by setting $\d(n_4,b,p_1)=n_5$
then $(p_0,p_1,n_4,n_5)$ becomes a fork with $n=n_4$, result $b
\in\out(n)$, $n_1 \act{p_0, a}
n_3 \act{p_0,a} n_4$ as $p_0$-path, and the $p_1$-path consisting
of the single node $n_5$.

Consider now the local circuit $n_2 \act{p_1, a} n_4 \act{p_1, b} n_2$
in the graph of the negotiation on the left of
Figure~\ref{fig:example1}. The node $n_4$  
is dominating, since its domain includes all processes; node $n_2$ is
not dominating since $p_0\not\in\dom(n_2)$.
\end{exa}

\begin{lemC}[{\cite[Lemma 2]{negII}}]
\label{lem:dominating}
Every reachable local circuit of a sound deterministic negotiation (that is, every local circuit
containing a node reachable from $n_\init$ by a local path) has a dominating node.\footnote{In \cite{negII} dominating nodes of circuits are called
  synchronizers.}
\end{lemC}

\begin{exa}
Lemma \ref{lem:dominating} does not hold for arbitrary sound negotiations. Consider the 
non-deterministic negotiation of Figure \ref{fig:nodom}. It is easy to see that the negotiation is sound.
However, the local circuit $n_1 \act{p_1, a} n_2 \act{p_1, a} n_1$ has no dominating node, because
$\dom(n_1) = \{p_0, p_1\}$ and $\dom(n_2) = \{p_1, p_2\}$.

\begin{figure}[htb]
\begin{tikzpicture}
\nego[ports=3,id=n0,spacing=2.0]{0,0}
\node[left = 0cm of n0_P0, font=\large] {$n_0$};
\nego[ports=2,id=n1,spacing=2]{0,-2}
\node[above right = -0.1cm and 0.5cm of n1_P0, font=\large] {$n_1$};
\nego[ports=2,id=n2,spacing=2]{2,-4}
\node[above right = -0.1cm and 0.5cm of n2_P0, font=\large] {$n_2$};
\nego[ports=3,id=nf,spacing=2]{0,-6}
\node[left = 0cm of nf_P0, font=\large] {$n_3$};

\node[above = 0.1cm of n0_P0, font=\large] {$p_0$};
\node[above = 0.1cm of n0_P1, font=\large] {$p_1$};
\node[above = 0.1cm of n0_P2, font=\large] {$p_2$};

\pgfsetarrowsend{latex}
\draw (n0_P0) to node[pos=0.47,auto] {$a$} (n1_P0);
\draw (n0_P1) to node[pos=0.47,auto] {$a$} (n1_P1);
\draw (n0_P2) to node[pos=0.2,auto] {$a$} (n2_P1);
\path (n0_P2) -- + (0, -1.5) edge [bend left = 60] (nf_P2);
\draw (n1_P0) to node[pos=0.5,auto] {$b$} (nf_P0);
\draw (n1_P0) to  node[pos=0.7,auto]{$a$} (-1.0,-2.4) to (nf_P0);
\draw (-1.0,-2.4) to[out=90,in=130, looseness=2] (n1_P0);
\draw (n1_P1) to[out=-140,in=140] node[pos=0.5,right,swap]{$b$} (nf_P1);
\draw (n1_P1) to[out=-110,in=110] node[pos=0.5,auto,swap]{$a$} (n2_P0);
\draw (n2_P0) to[out=70,in=-70] node[pos=0.5,auto,swap]{$a$} (n1_P1);
\draw (n2_P1) to  node[pos=0.5,below]{$a$} (5.0,-4.4) to (nf_P2);
\draw (5.0,-4.4) to[out=90,in=40, looseness=2] (n2_P1);
\draw (n2_P0) to node[pos=0.5,auto]{$b$} (nf_P1);
\draw (n2_P1) to node[pos=0.5,left]{$b$} (nf_P2);
\end{tikzpicture}
\caption{A local circuit without a dominating node}
\label{fig:nodom}
\end{figure}
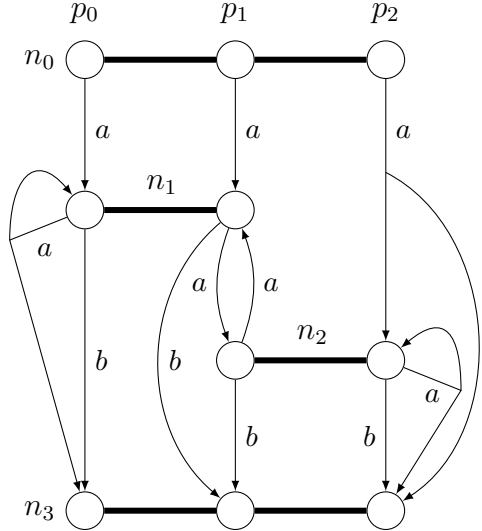
\end{exa}

\begin{defi}
\begin{enumerate}
\item An anti-pattern of type $\Bb$ is a $p$-path leading from $n_\init$ to a node $n$ such that no $p$-path
leads from $n$ to $n_\fin$.
\item An anti-pattern of type $\Ff$ is a fork $(p_1,p_2,n_1,n_2)$ such that $p_2 \in \dom(n_1)$ and $p_1 \in \dom(n_2)$.
\item An anti-pattern of type $\Cc$ is a local circuit without a dominating node. 
\end{enumerate}
\end{defi}

\begin{exa}
The last two anti-patterns are illustrated in Figure \ref{fig:anti}. The tuple
$(p_0, p_1, n_1, n_2)$ is a fork of the negotiation on the left satisfying $p_1 \in \dom(n_1)$ and $p_0 \in \dom(n_2)$. 
The local circuit $n_1 \act{p_0, a} n_2 \act{p_1, a} n_3 \act{p_2, a} n_1$ of the negotiation on the right 
has no dominating node. Observe that the negotiation has no
anti-pattern of type $\Ff$.

Another example of anti-pattern of type $\Ff$ appears in the
negotiation on the left of Figure~\ref{fig:example1} modified by  letting
$\d(n_4,b,p_1)=n_5$: the tuple $(p_0,p_1,n_4,n_5)$ is a fork with $p_0
\in\dom(n_5)$ and $p_1 \in\dom(n_4)$.

\begin{figure}[htb]
\raisebox{2cm}{\begin{tikzpicture}
\nego[ports=2,id=n0,spacing=1.5]{0,0}
\node[left = 0cm of n0_P0, font=\large] {$n_0$};
\nego[ports=2,id=n1, spacing=1.5]{-1.5,-2}
\node[left = 0cm of n1_P0, font=\large] {$n_1$};
\nego[ports=2,id=n2, spacing=1.5]{1.5,-2}
\node[right = 0cm of n2_P1, font=\large] {$n_2$};
\nego[ports=2,id=n3, spacing=1.5]{0,-4}
\node[left = 0.1cm of n3_P0, font=\large] {$n_3$};

\node[above = 0.1cm of n0_P0, font=\large] {$p_0$};
\node[above = 0.1cm of n0_P1, font=\large] {$p_1$};

\pgfsetarrowsend{latex}
\draw (n0_P0) edge node[left=0, pos=0.7] {$a$} (n1_P0);
\draw (n0_P0) edge node[right=0, pos=0.7] {$b$} (n2_P0);
\draw (n0_P1) edge node[left=0, pos=0.7] {$b$} (n1_P1);
\draw (n0_P1) edge node[right=0, pos=0.7] {$a$} (n2_P1);
\draw (n1_P0) edge node[left=0, pos=0.3] {$a$} (n3_P0);
\draw (n1_P1) edge node[left=0, pos=0.3] {$a$} (n3_P1);
\draw (n2_P0) edge node[right=0, pos=0.3] {$a$} (n3_P0);
\draw (n2_P1) edge node[right=0, pos=0.3] {$a$} (n3_P1);

\end{tikzpicture}}
\scalebox{0.9}{
\begin{tikzpicture}
\nego[ports=3,id=n0,spacing=2]{0,0}
\node[left = 0cm of n0_P0, font=\large] {$n_0$};
\nego[ports=2,id=n1, spacing=2]{-2,-2}
\node[right = 0cm of n1_P1, font=\large] {$n_1$};
\nego[ports=2,id=n2, spacing=2]{0,-4}
\node[right = 0cm of n2_P1, font=\large] {$n_2$};
\nego[ports=2,id=n3, spacing=2]{2,-6}
\node[right = 0.1cm of n3_P1, font=\large] {$n_3$};
\nego[ports=3,id=n4, spacing=2]{-2,-8}
\node[left = 0.1cm of n4_P0, font=\large] {$n_4$};

\node[above = 0.1cm of n0_P0, font=\large] {$p_0$};
\node[above = 0.1cm of n0_P1, font=\large] {$p_1$};
\node[above = 0.1cm of n0_P2, font=\large] {$p_2$};

\pgfsetarrowsend{latex}
\draw (n0_P0) edge node[left=0] {$a$} (n1_P1);
\draw (n0_P1) edge node[left=0] {$a$} (n2_P1);
\draw (n0_P2) edge node[right=0] {$a$} (n3_P1);
\draw (n1_P0) edge node[left=0] {$a$} (n4_P0);
\draw (n1_P1) edge node[left=0] {$a$} (n2_P0);
\draw (n2_P0) edge node[left=0] {$a$} (n4_P1);
\draw (n2_P1) edge node[left=0] {$a$} (n3_P0);
\draw (n3_P0) edge node[right=0, pos=0.7] {$a$} (n4_P2);
\draw (n3_P1) edge[out=-120,in=-70] node[right=0.2, pos=0.55] {$a$} (n1_P0);

\end{tikzpicture}}
\caption{Anti-patterns}
\label{fig:anti}
\end{figure}
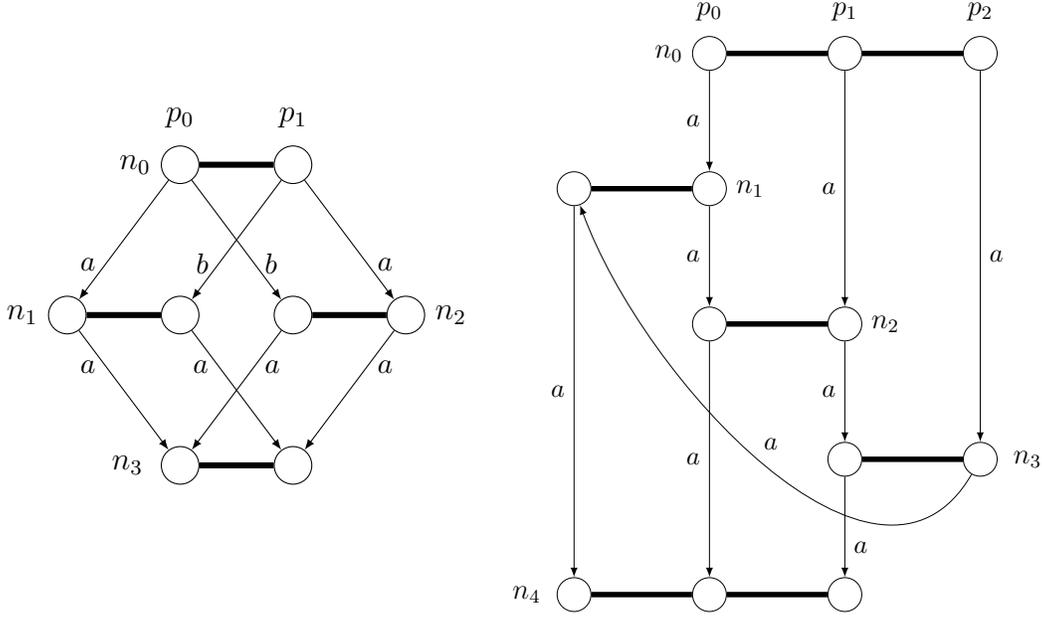
\end{exa}

\begin{lem}
Any deterministic negotiation containing some anti-pattern is unsound.
\end{lem}

\begin{proof} Assume that a sound, deterministic negotiation $\Nn$
contains an anti-pattern. If the anti-pattern is of type $\Bb$, then since the $p$-path leading to $n$
is realizable (Lemma \ref{l:local}), some reachable configuration $C$ satisfies $C(p) = \{n\}$. Since 
no $p$-path leads from $n$ to $n_\fin$, we have $C'(p) \neq \{n_\fin\}$ for every
configuration $C'$ reachable from $C$, and so $\Nn$ is not sound.

\smallskip

If the anti-pattern is of type $\Cc$, then the result follows from Lemma \ref{lem:dominating}.  

\smallskip

If the anti-pattern is of type $\Ff$, then $\Nn$ contains a fork $(p_1,p_2,n_1,n_2)$ such 
that $p_2 \in\dom(n_1)$, and $p_1 \in \dom(n_2)$. Let $n$ 
and $a$ be the node and the result required by the definition of a fork. Since local paths of sound  
deterministic negotiations are realizable (Lemma \ref{l:local}), some reachable configuration 
$C$ enables $n$. Let $C \act{(n,a)} C'$. By the definition of the anti-pattern, 
there are disjoint $p_1$- and $p_2$-paths $\pi_1$ and $\pi_2$ leading to 
$n_1$ and $n_2$, respectively. Using again soundness and the fact that
$\pi_1$ and $\pi_2$ are disjoint, we can show by induction on $|\pi_1|+|\pi_2|$ that there is 
a run  $C' \act{w} C''$ that realizes both $\pi_1$ and $\pi_2$  from $C'$. So we have $C''(p_1) = n_1$ and 
$C''(p_2)=n_2$, because neither $n_1$ can be executed before $n_2$
gets enabled, nor the other way round. Since $p_2 \in \dom(n_1)$ and $p_1 \in \dom(n_2)$, neither $n_1$ nor $n_2$ can ever occur from 
$C''$. By the same argument as above, this implies that the initial run leading to $C''$ cannot be extended to 
a successful run, and so $\Nn$ is unsound.
\end{proof}

\begin{lem}
\label{lem:comp}
Any unsound deterministic negotiation contains some anti-pattern.
\end{lem}
\begin{proof}
Let $\Nn$ be an unsound deterministic negotiation without anti-patterns of type $\Bb$.
We prove that $\Nn$ has some anti-pattern of type $\Ff$ or $\Cc$.
Let $\Proc=\{p_1, \ldots, p_n\}$ be the set of processes of $\Nn$.

By Lemma~\ref{lem:unsound-deadlock} we know that $\Nn$ has some
reachable deadlock configuration, let us call it $C$.
This means that there are processes 
$p_0,\ldots,p_{k-1}$, and distinct nodes $n_0,\ldots,n_{k-1},n_k=n_0$ ($k \ge
2$) such that   
$C(p_i)=n_i$ and $p_i \in\dom(n_{i+1})$ for every $0 \le i < k$.  
(Intuitively, at the configuration $C$ process $p_0$ waits for $p_{k-1}$, $p_1$ waits for $p_0$, etc.) 
A run from $C_\init$ to $C$ yields a fork 
$(p_i,p_{i+1},n_i,n_{i+1})$ for every pair of nodes $n_i,n_{i+1}$: Indeed, we can choose the pair 
$(n,a)$ required in the definition of a fork as the last pair in the 
run satisfying $\set{p_i,p_{i+1}} \subseteq \dom(n)$, and choose the path $\pi$ (resp. $\pi'$) 
as a $p_i$-path from $\d(n,a,p_i)$ to $n_i$ (resp. a $p_{i+1}$-path from $\d(n,a,p_{i+1})$ to $n_{i+1}$).

If $k=2$ then $(p_0,p_1,n_0,n_1)$ is a a pattern of type
$\Ff$. So assume that $k>2$ and that $\Nn$ has no pattern of type
$\Ff$.  We prove that $\Nn$ has a local circuit without a dominating node.

We claim that for every $0 \le i<k$  there is some $p_i$-path $\pi_i$ from
$n_i$ to $n_{i+1}$ (setting $n_k=n_0$). Assume the contrary, and let $(n,a)$ be the pair
in the definition of the fork $(p_i,p_{i+1},n_i,n_{i+1})$. Then we can
extend the $p_i$-path $\pi$  
from $\d(n,a,p_i)$ to $n_i$ with some $p_i$-path $\pi'$ from 
$n_i$ to some node $n'$ with $p_{i+1} \in\dom(n')$, and such that
$p_{i+1}$ does not occur in $\pi'$ except for $n'$. Such a node $n'$
exists because $\Nn$ has no anti-pattern of type $\Bb$, and the final node 
$n_\fin$ contains all processes. The
$p_i$-path $\pi \pi'$ is still disjoint with the $p_{i+1}$-path
leading from $\d(n,a,p_{i+1})$ to $n_{i+1}$.  Therefore, the fork 
$(p_i, p_{i+1}, n', n_{i+1})$ is an anti-pattern of type $\Ff$, 
contradicting the hypothesis, and the claim is proven.
  
Note that $p_{i+1} \notin \dom(n_i)$, since otherwise $(p_i,p_{i+1},n_i,n_{i+1})$ would be a pattern
of type $\Ff$. Also, $p_{i+1}$ cannot occur in $\pi_i$ except for
$n_{i+1}$: otherwise we could take the shortest prefix $\pi$ of
$\pi_i$ ending in a node $n'_i\not=n_{i+1}$ with $\set{p_i,p_{i+1}} \subseteq \dom(n'_i)$ and find another $\Ff$ pattern, 
using the paths $\pi_i \pi$ and $\pi_{i+1}$. Let $\gamma$ be the concatenation of  $\pi_0, \pi_1, \ldots \pi_{k-1}$. 
Then $\gamma$ is a local circuit of $\Nn$, and every $p_i$ belongs to the domain of some node of $\gamma$. Since  
$p_{i+1}$ does not belong to the domain of any node of $\pi_i$, except $n_{i+1}$, no node of $\gamma$ has all of 
$p_0, \ldots, p_{k-1}$ in its domain. So $\gamma$ has no dominating node, and therefore $\Nn$ contains a pattern of type $\Cc$.
\end{proof}

The above two lemmas allow us to prove:

\begin{thm}\label{th:sound-det}
 Soundness of deterministic negotiations is \NLOGSPACE-complete.
\end{thm}
\begin{proof}
The existence of an anti-pattern of type $\Bb$ or $\Ff$ is clearly verifiable in \NLOGSPACE.

To non-deterministically check the existence of a path of type $\Cc$ using logarithmic space, 
we guess an integer $M$ and a process $p$.
Then, we guess on-the fly a circuit $\pi$, such that (1) each node of
$\pi$ has domain of size at most $M$, (2) $p$ occurs in $\pi$, and (3)
one node of $\pi$ has domain of size $M$ and does not contain $p$. Clearly,
this implies that $\pi$ contains at least $M+1$ processes, but no
dominant node.

We prove \NLOGSPACE-hardness by reduction from the reachability problem for
directed graphs, which is \NLOGSPACE-complete. Given a directed graph $G$ with 
nodes $s$ and $t$, where $s$ has no incoming and $t$ no outgoing edges, we reduce
the question whether $s \act{*} t$ to the soundness of a deterministic 
negotiation $\Nn(G)$ with one process $p$.  The atomic negotiations of $\Nn(G)$ are the 
vertices of $G$, with domain $\{p\}$. The initial node is $s$ and the final node is $t$.
If $u \rightarrow v$ is an edge of $G$, then the atomic 
negotiation $u$ has a result $(u, v)$ with 
$\delta(u,(u,v),p) = v$.  Moreover, for every vertex $u$ of $G$ except
$t$ the atomic negotiation $u$ has a result ${\it back}$ with
$\delta(u, {\it back}, p) = s$.  Clearly, $\Nn(G)$ is sound
if{}f there is some path $s \act{*} t$ in~$G$. 
\end{proof}

\begin{rem}
\label{rem:soundequiv}
As announced in Remark \ref{rem:sounddef},
for deterministic negotiations the 
notion of soundness used in this paper and the one of \cite{negI} 
essentially coincide. More precisely, we show that the two notions coincide
under the very weak assumption that for every atomic negotiation $n$ there
is a local path from $n_\init$ to $n$. (Atomic negotiations that do not 
satisfy this condition can be identified and removed in \NLOGSPACE, and their 
removal does not change the behavior of the negotiation.)

Recall that the definition of \cite{negI} requires that 
(a) every run can be extended to a successful run \emph{and} (b) that for 
every atomic negotiation $n$ some initial run enables $n$. Now, let 
$\Nn$ be a deterministic negotiation that is sound according to the 
definition of this paper, i.e., satisfies (a), and such that for every 
atomic negotiation $n$ there is a local path from $n_\init$ to $n$.
We show that (b) also holds. For $n = n_\fin$, (b) holds because 
every run can be extended to a successful run. For $n \neq n_\fin$, we observe 
that, by Lemma \ref{l:local}, the local path leading from $n_\init$ to $n$ 
is realizable. By the definition of realizability, there is a reachable 
configuration $C$ and a process $p$ such that $C(p) = \{n\}$. Since $\Nn$ 
is sound, it has a run leading from $C$ to $C_\fin$. 
Since $\N$ is deterministic and $n \neq n_\fin$, this run necessarily 
executes $n$, and we are done.
\end{rem}


\section{Beyond determinism: Tractable cases}
\label{sec:beyond-det-tract}

In this section we investigate how much nondeterminism we can allow,
while retaining polynomiality of the soundness problem, a question that was 
left open in \cite{negI,negII}.  We prove that soundness of acyclic, 
weakly non-deterministic negotiations can be decided in polynomial time. 
In Section \ref{sec:beyond-det-intract} we show that the problem becomes intractable for 
both arbitrary  weakly non-deterministic negotiations, and for acyclic non-deterministic negotiations.  

The polynomial time algorithm is based on a game-theoretic solution to 
the \emph{omitting problem} for deterministic negotiations, which is a
problem of independent interest. 
Section \ref{sec:omit} introduces an extended version of the omitting
problem for deterministic negotiations where we not only want to omit
some nodes but also to reach 
some other nodes. We show that for a fixed number of nodes to be
reached, the problem can be solved  
in polynomial time for sound, acyclic and deterministic negotiations. Section~\ref{sec:weak} uses this result to prove that
the soundness problem for acyclic weakly non-deterministic negotiations is in \PTIME.
\igwin{stressed that omitting is only for deterministic negotiations.}

\subsection{Omitting problem}\label{sec:omit}

Let $B\incl N$ be a set of nodes of a deterministic negotiation $\Nn$. 
We say that a run $(n_1,a_1)(n_2,a_2)\cdots$  of $\Nn$
\emph{omits $B$} if $n_i\not\in B$ for all $i$.
Let $P\incl N\times R$ be a set of pairs consisting of a node and a result.
We say that a run of $\Nn$ \emph{includes $P$ and omits $B$} if it
omits $B$ and contains all the pairs from $P$.

\begin{defi}
The \emph{omitting problem} consists of deciding, given $\Nn$, $P$, 
and $B$, whether there is a successful run of $\Nn$ including $P$ and omitting
$B$. 

Given a constant $K$, the \emph{$K$-omitting problem} is the subproblem
of the omitting problem in which $P$ has size at most $K$.
\end{defi}

We show that the omitting problem for  sound, acyclic, and
deterministic negotiations can be reduced to solving a safety game on
a finite graph
(see e.g.~\cite{Thomas08} for an introduction to games). 
As a first step we  define a \emph{two-player game $\GNB$} with players Adam and Eve, 
where the goal of Eve is to produce a successful run that omits
$B$. The game arena is a finite graph, with vertex set partitioned
between Adam and Eve:
\begin{itemize}
\item the positions of Eve are $N\setminus B$, namely the set of
  all the nodes of the negotiation but those in $B$;
\item the positions of Adam are $N\times R$, namely the set of
  pairs consisting of a node and an action
\end{itemize}
The edges of the game arena are:
\begin{itemize}
\item $n \act{} (n,a)$, for every $n \in N\setminus B$, $a\in \out(n)$;
\item $(n,a) \act{} \d(n,a,p)$, for every $n \in N$, $a \in \out(n)$,
  $p\in \dom(n)$.
\end{itemize}
The initial game position is $n_\init$. Adam wins a play if it reaches a node in $B$, Eve wins if the play
  reaches $n_\fin$. Eve has a winning strategy in the game, if all
  plays (from $n_\init$) are won by her, no matter which are the moves of Adam. 

The idea behind the game is that Eve chooses results, and builds thus
a run of the negotiation. Adam challenges Eve by choosing at each step
a process, thus a play is a local path of the negotiation. The winning
condition for Eve ensures that the execution is successful. Moreover,
the execution avoids $B$ since Eve is not allowed to continue from nodes in $B$. 

Observe that since $\Nn$ is acyclic, the winning condition
for Eve is actually a safety condition: every maximal play avoiding
$B$ is winning for Eve. 
So if Eve  can win, then she wins with a positional strategy.
A \emph{deterministic positional strategy for Eve} is a function $\s:N\to R$, it
indicates that at position $n$ Eve should go to position
$(n,\s(n))$.
Since $\GNB$ is a safety game for Eve, there is a \emph{biggest
non-deterministic winning strategy} for Eve, 
i.e., a strategy of type $\smax:N\to\Pp(R)$. 
The strategy $\smax$ is obtained by computing the set $W_E$ of all
winning positions for Eve in $\GNB$, and then setting for every
$n\in N$:
\begin{equation*}
  \smax(n)=\set{a\in\out(n) : \text{ for all $p\in\dom(n)$, } \d(n,a,p)\in W_E}\,.
\end{equation*}

\begin{lem}\label{lemma:run-to-strat}
  Let $\Nn$ be an acyclic negotiation. If $\Nn$ has a successful run omitting $B$ then Eve has a winning strategy in $\GNB$.
\end{lem}

\begin{proof}
Define $\s(n)=a$ if $(n,a)$ appears in the run. For other nodes
define the strategy arbitrarily. 
To check that this strategy is winning, it is enough to verify that
every play respecting the strategy stays in the nodes appearing in the
run. 
\end{proof}

\begin{lem}\label{lemma:strat-to-run}
  Let $\Nn$ be a sound, acyclic, and deterministic negotiation such that 
  Eve has a deterministic winning strategy $\s:N \to R$ in
  $\GNB$.  Consider the set $S$ of nodes that are reachable on some play
  from $n_\init$ respecting $\s$. Then there exists a successful run of $\Nn$ containing
  precisely the nodes of $S$.
\end{lem}
{
\begin{proof}
Recall that for an acyclic $\Nn$ we can fix a topological order
$\leqN$. Let $n_1,n_2,...,n_k$ be an enumeration of the nodes in $S\subseteq(N \setminus B)$
according to $\leqN$.
Let $w_i=(n_1,\s(n_1))\cdots(n_i,\s(n_i))$.
By induction on $i\in \set{1,\dots,k}$ we prove that 
there is a configuration $C_i$ such that $C_\init\act{w_i} C_i$ is a
run of $\Nn$. 
This will show that $w_k$ is a successful run containing precisely the
nodes of $S$.

For $i=1$, $n_1=n_\init$, in $C_\init$ all
processes are ready to do $n_1$, so $C_1$ is the result of performing
$(n_1,\s(n_1))$. 

For the inductive step, we assume that we have a run
$C_\init\act{w_i} C_i$, and we want to extend it by
$C_i\act{(n_{i+1},\s(n_{i+1}))} C_{i+1}$. 
Consider a play respecting $\s$ and reaching $n_{i+1}$. 
The last step in this play is $(n_j,\s(n_j))\to n_{i+1}$, for some
$j \le i$ and $n_j$ in $S$. 
Since $\Nn$ is deterministic, we have $\d(n_j,\s(n_j),p)=n_{i+1}$ for some process $p$. 
Since $j \le i$ and $(n_j,\s(n_j))$ occurred in $w_i$ (but 
$n_{i+1}$ has not occured), we have
$C_i(p) =\set{n_{i+1}}$.
If we show that $C_i(q)=\set{n_{i+1}}$ for all $q\in\dom(n_{i+1})$ then we
obtain that $n_{i+1}$ is enabled in $C_i$ and we get the required  $C_{i+1}$.
Suppose by contradiction that $C_i(q)=\set{n_l}$ for some
$l \not= i+1$. We must have $l>i+1$, since
otherwise $n_l$ already occurred in $w_i$.
By definition of our indexing $n_{i+1}\leN n_l$.
But then no run from $C_i$ can bring process $q$ to a
state where it is ready to participate in negotiation $n_{i+1}$, and so, since $\Nn$ is deterministic,
we have $C(p) = \{n_{i+1} \}$ for every configuration $C$ reachable from $C_i$. 
This contradicts the fact that $\Nn$ is sound.
\end{proof}
}

Lemmas~\ref{lemma:run-to-strat} and~\ref{lemma:strat-to-run} show that
Eve wins in $\GNB$ iff $\Nn$ has a successful run omitting $B$. We apply this result
to the omitting problem.

\begin{thm}\label{thm:omitting}
For every constant $K$, the $K$-omitting problem for 
deterministic, acyclic, and sound negotiations is in \PTIME.
\end{thm}
{
\begin{proof}
  If for some atomic negotiation $m$ we have $(m,a)\in P$ and
  $(m,b)\in P$ for $a\not =b $ then the answer is negative as $\Nn$
  is acyclic.
  So let us suppose that it is not the case.
  By Lemmas~\ref{lemma:run-to-strat} and~\ref{lemma:strat-to-run} our
problem is equivalent to determining the existence of a deterministic
strategy $\s$ for Eve in the game $\GNB$ such that $\s(m)=a$ for all
$(m,a)\in P$, and all these $(m,a)$ are reachable on a play respecting
$\s$.

  To decide this we calculate $\smax$, the biggest non-deterministic winning
  strategy for Eve in $\GNB$. 
  This can be done in \PTIME\ as the size of $\GNB$ is proportional to
  the size of the negotiation.
  Strategy $\smax$ defines a graph $G(\smax)$ whose nodes are atomic
  negotiations, and edges are $(m,a,m')$ if $(m,a)\in\smax$ and
  $m'=\d(m,a,p)$ for some process $p$.
  The size of this graph is proportional to the size of the negotiation.
  In this graph we look for a subgraph $H$ such that:
  \begin{itemize}
  \item for every node $m$ in $H$ there is at most one $a$ such that
    $(m,a,m')$ is an edge of $H$ for some $m'$;
  \item for every $(m,a)\in P$ there is an edge $(m,a,m')$ in $H$ for
    some $m'$, and moreover $m$ is reachable from $n_\init$ in $H$.
  \end{itemize} 
  We show that  such a graph $H$ exists iff there is a strategy $\s$
  with the required properties.

  Suppose there is a deterministic winning strategy
  $\s$ such that $\s(m)=a$ for all $(m,a)\in P$, and all these $(m,a)$
  are reachable on a play respecting $\s$. 
  We now define $H$ by putting an edge $(m,a,m')$ in $H$ if $\s(m)=a$
  and 
  $m'=\d(m,a,p)$ for some process $p$. 
  As $\s$ is deterministic and winning, this definition guarantees that $H$
  satisfies the first item above. 
  The second item is guaranteed by the reachability property that
  $\s$ satisfies. 
  
  For the other direction, given such a graph $H$ we define a
  deterministic strategy  $\s_H$. 
  We put $\s_H(m)=a$ if $(m,a,m')$ is an edge of $H$. 
  If $m$ is not a node in $H$, or has no outgoing edges in $H$ then we
  put $\s_H(m)=b$ for some arbitrary $b\in \smax(m)$.
  It should be clear that $\s_H$ is winning since every play
  respecting $\s_H$ stays in the winning nodes for Eve. 
  By definition $\s_H(m)=a$ for all $(m,a)\in P$, and all these
  $(m,a)$ are reachable on a play respecting $\s_H$.

  So we have reduced the problem stated in the theorem to finding a
  subgraph $H$ of $G(\smax)$ as described above. 
  If there is such a subgraph $H$ then there is one in form of a tree,
  where the edges leading to leaves are of the form  $(m,a,m')$  with
  $(m,a)\in P$. 
  Moreover, there is such a tree with a few, at most $|P|$, nodes with more
  than one child.
  So finding such a tree can be done by guessing the $|P|$ branching nodes and
  solving $|P|+1$ reachability problems in $G(\smax)$. 
  This can be done in \PTIME\ since the size of $P$ is bounded by $K$.
  \end{proof}
}

\subsection{Soundness of acyclic weakly non-deterministic negotiations}\label{sec:weak}

In this section we consider acyclic, weakly non-deterministic
negotiations, c.f.\ page~\pageref{def:weak}. 
That is, we allow some processes to be non-deterministic, but every
atomic negotiation should involve at least one  deterministic process.
We will often work with a part of a negotiation as defined below.
\begin{defi}
The~\emph{restriction} of a negotiation $\Nn= \struct{\Proc,N,\dom,R,\d}$ 
to a subset $\Proc' \subseteq \Proc$ of its processes is the negotiation
$\struct{\Proc',N',\dom',R,\d'}$ where
$N' = \{ n\in N \colon \dom(n)\cap\Proc'\not=\es\}$, 
$\dom'(n)=\dom(n)\cap\Proc'$,
and $\d'(n,r,p)=\d(n,r,p)\cap N'$. The restriction of $\Nn$ to
its deterministic processes is denoted $\Nn_D$.
\end{defi}

If $\Nn$ is weakly non-deterministic, every atomic negotiation
involves a deterministic process, so $\Nn_D=\Nn$ have the same set of
nodes; but they do not have the same set of processes unless $\Nn$ is
deterministic. 

Recall also that for an acyclic negotiation $\Nn$ we fixed some linear order
$\leqN$ on $N$, that is a topological order of the graph of $\Nn$.

We show that deciding soundness for acyclic, weakly non-deterministic negotiations is in $\PTIME$.
The proof is divided into three parts. In Section \ref{subsub:prelim} we prove some preliminary lemmas. 
In Section \ref{subsub:one} we consider
the special case in which the negotiation $\Nn$ has one single non-deterministic process, and show that
$\Nn$ is sound if{}f $\Nnd$ is sound and Eve wins a certain instance of the omitting problem for $\Nnd$.
Finally, Section \ref{subsub:more} shows how to reduce the general case to the case with only
one non-deterministic process.

\subsubsection{Preliminaries}
\label{subsub:prelim}
We first show two auxiliary lemmas on the structure of runs in
negotiations. Two runs $w,w'$ of $\Nn$ are called \emph{equivalent}
(and we write $w \equiv w'$) if one can be obtained from the other by
repeatedly permuting adjacent pairs $(m,a),(n,b)$, with $\dom(m)\cap
\dom(n)=\es$. 

The next lemma shows that for acyclic negotiations we can restrict our considerations to runs respecting the order $\leqN$.

\begin{lem}\label{lemma:order}
  Every run of an acyclic negotiation $\Nn$ has an equivalent run that
  respects the  topological order $\leqN$. 
\end{lem}
\begin{proof}
  Let $C_\init \act{(n_0,a_0) \cdots (n_p,a_p)} C$ be some run of $\Nn$. 
  The proof is by
  induction on the number of pairs $0 \le i < j \le p$ such that $n_j
  \leqN n_i$. If there are no such pairs then we are done. Otherwise,
  assume that a pair $i,j$ as above exists. Note that in particular,
  $\dom(n_i) \cap \dom(n_j)=\es$ holds. It is not hard to see that
  there must exist some $i \le k < j$ such that $n_{k+1}  \leqN n_k$. As
  before, we note that $\dom(n_k) \cap
  \dom(n_{k+1})=\es$. Clearly, $(n_0,a_0) \cdots
  (n_{k-1},a_{k-1})(n_{k+1},a_{k+1}) (n_k,a_k)  (n_{k+2},a_{k+2})
  \cdots (n_p,a_p)$ is an equivalent run of $\Nn$, with one less pair
  that violates $\leqN$.
\end{proof}

The next lemma states some properties of runs respecting the order $\leqN$.

\begin{lem} \label{lem:small}
Let $\Nn$ be an acyclic, weakly non-deterministic negotiation.
Consider some reachable
configuration $C$. Let $n$ be the $\leqN$-smallest atomic negotiation with $C(d)=n$ for some deterministic process $d$. 
The following properties hold:
\begin{enumerate}
\item If $\Nn$ is sound then $n$ is enabled in $C$.
\item For all runs $C \act{w}C_\fin$, all atomic negotiations $n'$ in $w$ are $\leqN$-bigger than $n$. 
\end{enumerate}	
\end{lem}

\begin{proof}
For the first item we use the soundness of $\Nn$: if $n$ were not enabled in $C$, then there 
would exist some $n'$ that is enabled in $C$ and such that the execution of $n'$ makes the execution
of $n$ eventually possible. So in particular, we would have $n' \leN n$, which contradicts the choice of $n$.

	The second item is shown similarly: every atomic negotiation $n'$ occurring in $w$
satisfies $n'' \leqN n'$ for some $n''$ that is enabled in $C$. By the choice of $n$ we have $n \leqN n''$, 
which shows the claim.
\end{proof}

It is easy to see that whenever $C$ is a reachable configuration of a weakly non-deterministic negotiation $\Nn$, the restriction 
$C^D$ of $C$ to deterministic processes  is a reachable configuration
of $\Nn_D$. The next lemma considers the reverse construction, lifting runs of
$\Nn_D$ to runs of $\Nn$. For this, we need to assume that the runs respect the order $\leqN$.

\begin{lem}\label{lemma:from-det-to-nondet}
  Suppose $\Nn$ is a sound, acyclic and weakly non-deterministic
  negotiation. Let $\Cinit^D\act{w} C^D$ be a run of $\Nnd$ respecting
  the order $\leqN$. Let also $C_v^D$ denote the configuration reached by some prefix $v$ of $w$:
  $\Cinit^D \act{v} C_v^D$. Then $\Nn$ has a run $\Cinit \act{v} C_v$
  with $C_v(d)=C_v^D(d)$, for every prefix $v$ of $w$ and every deterministic 
  process $d$.
\end{lem}

\begin{proof}
  The proof is by induction on the length of $v$.
  Consider a prefix $v(m,a)$ of $w$. 
  Recall that $w$ respects the order $\leqN$, thus  $m$ is the $\leqN$-smallest atomic negotiation enabled
  in $C^D_v$. In particular, Lemma~\ref{lem:small} (2) applied to $\Nn_D$ implies that $m$ is the 
  $\leqN$-smallest atomic negotiation with $C_v(d)=m$ for some deterministic process. Applying Lemma~\ref{lem:small} (1) to 
  $\Nn$ shows finally that $m$ is also enabled in $C_v$.  
\end{proof}

The next lemma gives a necessary condition for the
soundness of $\Nn$ that is easy to check. It is proved by showing that $\Nnd$ cannot have
much more behaviors than $\Nn$.  
\begin{lem}\label{lemma:sound-then-detsound}
    If $\Nn$ is a sound, acyclic, weakly non-deterministic negotiation
    then $\Nnd$ is sound.
\end{lem}
{
\begin{proof}
  Suppose to the contrary that $\Nnd$ has a run reaching a deadlock
  configuration $\Cinit\act{w} C$.
  By Lemma~\ref{lemma:order} we can assume that $w$ respects the $\leqN$
  ordering.
  By Lemma~\ref{lemma:from-det-to-nondet} we get a run of $\Nn$ to a
  deadlock configuration, but this is impossible.
  \end{proof}
}

\subsubsection{Negotiations with one non-deterministic process}
\label{subsub:one}

We  consider a special case of a negotiation with only one
non-deterministic process. The next lemma establishes the connection to the omitting problem: $\Nn$ is unsound
if the deterministic part $\Nnd$ is unsound, or $\Nnd$ is sound but has a certain successful omitting run.

\begin{lem}\label{lem:1nd} Let $\Nn$ be an acyclic, weakly
non-deterministic negotiation with a single non-deterministic process
$p$. Then $\Nn$ is not sound if and
only if either:
\begin{itemize}
\item $\Nn_D$ is not sound, or
\item $\Nn_D$ is sound, and it has  two nodes $m\leqN n$ with  results $a\in\out(m)$, $b\in\out(n)$ such that:
  \begin{itemize}
    \item $p\in\dom(m)\cap\dom(n)$,  $n\not\in \d(m,a,p)$, and
    \item there is a successful run of $\Nnd$ containing
      $P=\set{(m,a),(n,b)}$ and omitting
      $B=\set{n'\in \d(m,a,p) : m\leN n' \leN n}$.
  \end{itemize}
\end{itemize}
\end{lem}
{
\begin{proof}
  Consider the right-to-left direction.  We abbreviate $S_p:=\d(m,a,p)$. 
  If $\Nnd$ is not sound then by
  Lemma~\ref{lemma:sound-then-detsound}, $\Nn$ is not sound.

  Suppose then that $\Nn_D$ satisfies the second item from the
  statement of the lemma, and take a run $w$ of $\Nnd$ as it is
  assumed there. 
  By Lemma~\ref{lemma:order} we can assume that this run respects
  $\leqN$.
  Towards a contradiction suppose also that $\Nn$ is sound.
  Lemma~\ref{lemma:from-det-to-nondet} says that $w$ is also a run of
  $\Nn$. 
  Let $C_1$ be the configuration of this run just after $(m,a)$  was executed, 
  so we have  $C_1(p)=S_p$.
  Let $C_2$ be the first configuration after $C_1$ such that $C_2(d)=n$
  for some process $d$ (it may be that $C_2=C_1$).
  We have  $C_2(p)=S_p$ since the run $w$ omits $\set{n'\in S_p : m\leNd n'
    \leNd n}$, so $p$ cannot move between $C_1$ and $C_2$.
  When we continue following $w$ from $C_2$ we see that 
  either $p$ will never move, or it will move to some $n'\not=n$ with
  $n\leNd n'$. But then $d$ will not be able to move. 
  So this run leads to a deadlock, contradiction with the soundness
  of $\Nn$.

  For the left-to-right direction, assume that $\Nn_D$ is sound. We
  need to show the second item of the lemma.  
  Observe that since $\Nn$ is acyclic and not sound, there is a
  run $C_\init\act{w} C$ where $C$ is a deadlock.   
  By Lemma~\ref{lemma:order} we can assume that $w$ respects $\leqN$.
  Let $n$ be the $\leqN$-smallest  atomic negotiation such that
  $C(d)=n$ for some deterministic process $d$. 
  Applying Lemma~\ref{lem:small} (1) to $\Nn_D$ yields that  
  $n$ must be deterministically enabled in $C$: that is $C(d')=n$ for
  all deterministic processes $d'\in\dom(n)$. 
  This implies $p\in \dom(n)$, and $n\not\in C(p)$.

  Let us split $w$ as $u(m,a)v$ where $p\in \dom(m)$ and all
  atomic negotiations in $v$ involving $p$ are $\leqN$-bigger than $n$ 
  (it may be that $v$ is empty). That is to say, we define $m$ as the last atomic negotiation in $w$ involving $p$ that is $\leqN$-smaller than $n$.
  Let $C_m$ be the configuration reached after doing $(m,a)$:
  $\Cinit\act{u(m,a)} C_m$.
  Take $S_p=C_m(p)=\d(m,a,p)$. 
  By the choice of $m$ and $v$, the run from $C_m$ to $C$ does not use 
  any atomic negotiation from the set
  $\set{n'\in S_p :  m\leN n' \leN n}$.
  
  By soundness of $\Nnd$, from $C^D$ there is a run to the final
  configuration, and by the choice of $n$ and Lemma~\ref{lem:small} (2) this run cannot use 
  any atomic negotiation that is $\leqN$-smaller than $n$. 
  Let $b$ be the result such that $(n,b)$ appears in this run.
  Putting these pieces together we have a successful run of $\Nnd$
  containing $\set{(m,a),(n,b)}$ 
  and omitting $\set{n'\in S_p :  m\leN n' \leN n}$.
  We have already observed that  $p\in\dom(m)\cap\dom(n)$ and
  $n\not\in S_p=\d(m,a,p)$. So all the requirements of the lemma are met.
\end{proof}
}

\SKIP{
\begin{lem}\label{lemma:ptime-test}
  Soundness of acyclic, weakly
non-deterministic negotiations with only one non-deterministic process
can be checked in \PTIME.
\end{lem}

\begin{proof} 
For every $m\leqN n$, $a$ and $b$ we check the conditions described in
Lemma~\ref{lem:1nd}. 
The existence of a run of $\Nnd$ can be checked in \PTIME\ thanks to
Theorem~\ref{thm:omitting} and the fact that the size of $P$ is always
$2$. 
\end{proof}
}



\subsubsection{General weakly non-deterministic negotiations}
\label{subsub:more}
The next lemma deals with the case where there is more than one
non-deterministic process. Loosely speaking, in this case $\Nn$ is unsound 
if{}f there is a non-deterministic process such that the restriction of $\Nn$
to the deterministic processes \emph{and} this process is unsound.

\begin{lem}\label{l:vw-oneproc}
  An acyclic, weakly non-deterministic negotiation $\Nn$ is unsound if{}f:
  \begin{enumerate}
  \item its restriction $\Nnd$ to deterministic processes is unsound, or
\item for some non-deterministic process $p$, its restriction $\Nn^p$ to $p$
  and the deterministic processes is unsound.
  \end{enumerate}
\end{lem}
\begin{proof}
 For the right-to-left direction the case where $\Nnd$ is unsound
 follows directly from Lemma~\ref{lemma:sound-then-detsound}.
 It remains to check the case where $\Nn^p$ is not sound for some
 non-deterministic process $p$.
 Consider a run $C^p_\init\act{w} C^p$ where $C^p$ is a deadlock.
 By Lemma~\ref{lemma:order} we can assume that $w$ respects $\leqN$.
 Now let us try to make $\Nn$ execute the sequence $w$. 

 If $C_\init\act{w} C$ is a run of $\Nn$ then $C$ is a deadlock.
 Indeed if in $\Nn$ it would be possible to do $C\act{(n,a)} $ for
 some $(n,a)$ then $C^p\act{(n,a)}$ would be possible in $\Nn^p$. 

 The other case is when in $\Nn$ it is not possible to execute all the
 sequence $w$.
 Then we have $w=v(n,a)v'$, a run $C_\init\act{v}C_1$ and from $C_1$
 action $(n,a)$ is not possible.
 Since $C^p_\init\act{v}C^p_1\act{(n,a)}C^p_2$ is a run of $\Nn^p$, we
 know from Lemma~\ref{lem:small} (2) that there is a deterministic process $d$ with $C_1(d)=n$, and
 $n\leqN C_1(d')$ for all other deterministic processes $d'$. 
 Thus $C_1$ is a deadlock because by Lemma~\ref{lem:small} (1) if
there is an action possible from $C_1$ then this must be $n$.

For the left-to-right direction, suppose that $\Nn$ is not sound and
take a run $\Cinit\act{w} C$ with $C$ a deadlock configuration. 
  Take the $\leqN$-smallest atomic negotiation $n$ such
  that $n=C(d)$ for some deterministic process $d$. 
  Consider an arbitrary non-deterministic process $p$ and a run $\Cinit^p\act{w}
  C^p$ of $\Nn^p$; it is indeed a run since $\Nn^p$ is a restriction
  of $\Nn$. 
  As $\Nn^p$ is sound, it is possible to extend this run. 
  By Lemma~\ref{lem:small} (1), it should be possible to execute $n$
  from $C^p$. 
  Hence for every deterministic process $d\in\dom(n)$, we have
  $C(d)=n$.
  Moreover $n\in C(p)$ if $p\in \dom(n)$ is non-deterministic.
  Since the choice of $p$ was arbitrary, we have $n\in C(p)$ for all
  $p\in\dom(n)$. 
  Thus it is possible to execute $n$ from $C$, a contradiction. 
\end{proof}

\begin{thm}\label{th:weak-sound}
Soundness can be decided in $\PTIME$ for acyclic, weakly non-deterministic negotiations.
\end{thm}
\begin{proof}
  By Lemma~\ref{l:vw-oneproc} we can restrict to negotiations $\Nn$ with one
  non-deterministic process.
For every $m\leqN n$, $a$ and $b$ we check the conditions described in
Lemma~\ref{lem:1nd}. 
The existence of a run of $\Nnd$ can be checked in \PTIME\ thanks to
Theorem~\ref{thm:omitting} and the fact that the size of $P$ is $2$.
%
\end{proof}


\section{Beyond determinism: Intractable cases}
\label{sec:beyond-det-intract}
We show that if we remove any of the two assumptions of Theorem \ref{th:weak-sound} (acyclicity and weak non-determinism) then the soundness problem becomes
\coNP-complete. In fact, even a very mild relaxation of acyclicity
suffices.

It is not very surprising that deciding soundness for acyclic, 
non-deterministic negotiations is \coNP-complete. In the acyclic case
the negotiation is unsound if{}f it can reach a deadlock. Clearly,
by acyclicity this would happen after at most $|N|$ steps.  So it suffices to guess a run of linear length
step by step and check if it leads to a deadlock. \coNP-hardness is
proved by a simple reduction of SAT to the complement of the soundness problem. 
The reduction, which strongly relies on non-determinism, is 
presented in \cite{negI}. We get:

\begin{prop}[\cite{negI}]
  Soundness of acyclic non-deterministic negotiations is
  \coNP-complete.
\end{prop}

Now we consider a very mild relaxation of acyclicity: deterministic processes still need to be acyclic, but non-deterministic processes may have cycles.

Recall that $\Nnd$ is the restriction of $\Nn$ to deterministic processes.
\begin{defi}
A negotiation $\Nn$ is \emph{det-acyclic} if $\Nn_D$ is acyclic.
\end{defi}

It follows easily from this definition that all runs 
of a weakly non-deterministic, det-acyclic negotiation have length at
most $|N|$. 
However, we show that even in the \emph{very} weakly non-deterministic case (c.f. page~\pageref{def:very-weak}) the soundness problem is \coNP-complete. 



\begin{thm}\label{t:nphard}
  Non-soundness of det-acyclic, very weakly non-deterministic
  negotiations is \NP-complete.
\end{thm}

\begin{proof}
  We describe a reduction from 3-SAT and fix a 3-CNF formula
  $\varphi=c_1\wedge\dots\wedge c_m$, with
  clauses $c_1,\ldots,c_m$, each of length 3, and $k$ variables
  $x_1,\ldots, x_k$. Let $c_j=\ell_{j,1} \vee \ell_{j,2} \vee
  \ell_{j,3}$, with $\ell_{j,d} \in \set{x_i,\overline{x_i} : 1 \le i
    \le k}$ for every $1 \le j \le m$, $d\in \set{1,2,3}$. We construct a det-acyclic, very weakly non-deterministic
  negotiation $\Nn$ such that $\varphi$ is satisfiable if{}f $\Nn$ is not sound.

  The atomic negotiations of $\Nn$ are (apart from $n_\init$ and $n_\fin$):
  \begin{itemize}
  \item An atomic negotiation $m_0$, and for each variable $x_i$ three atomic negotiations $n^+_i, n^-_i, m_i$. 
  \item For every pair clause/literal $(j,d)$, $j=1,\dots, m$ and $d=1,2,3$, three atomic negotiations
              $m_{j,d}$, $n_{j,d}$, and $r_{j,d}$. 
  \item For every clause $c_j$, two auxiliary atomic negotiations $t_j$ and $t_j'$.
  \end{itemize}
  The processes of $\Nn$ are:
  \begin{itemize}
  \item A deterministic process $E$. 
  \item For each clause $c_j$, a deterministic process $V_j$.
  \item For every pair clause/literal $(j,d)$, where $j=1,\dots, m$ and $d=1,2,3$: two deterministic processes 
        $T_{j,d}, T'_{j,d}$, and a non-deterministic process $P_{j,d}$.
  \end{itemize}

 Now we describe the behavior of each process $P$ by means of a graph. The
 nodes of the graph for $P$ are the atomic negotiations in which $P$ participates. 
 The graph has an edge $n \rightarrow n'$ if
 there is a result $a$ of $n$ such that $P$ moves with $a$ from $n$ to $n'$.
 If $P$ is nondeterministic, and after $a$ is ready to engage in a set of atomic negotiations
 $\{n_1, \ldots, n_k\}$, then the graph contains a {\em hyperarc} leading from $n$ to $\{n_1, \ldots, n_k\}$.

The graphs of all processes are shown in Figure \ref{fig:processes}.
Intuitively, process $E$ is in charge of producing a valuation of $x_1,\ldots, x_k$:
it chooses between $n^+_1$ and $n^-_1$, then between $n^+_2$ and $n^-_2$, etc. Choosing 
$n^+_i$ stands for setting $x_i$ to true, and choosing $n^-_i$ for setting $x_i$ to false.

   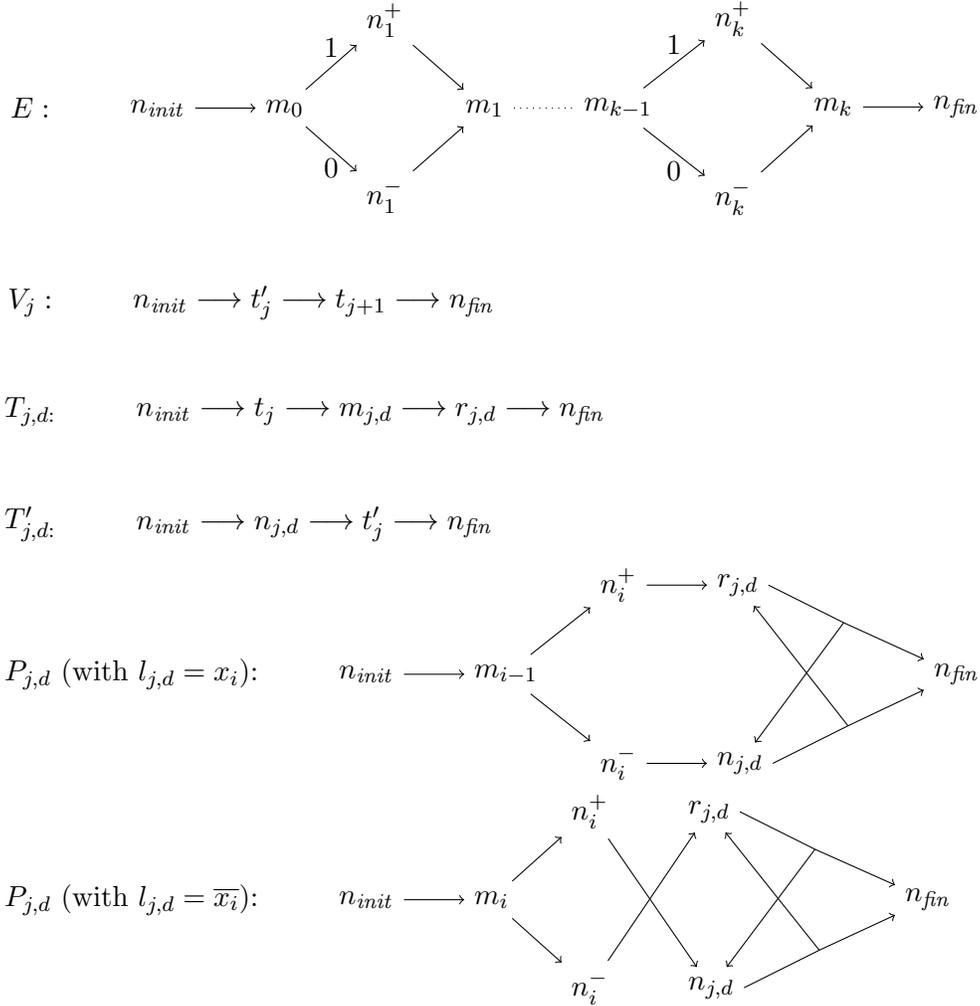
\begin{figure}[h]
     \centering
\begin{tikzpicture}[node distance=.8cm]

\node (init) {$n_\init$};
\node[right= of init] (q1)  {$m_0$};
\node[left= of init] (E) {$E:$};
\node[above right= of q1] (n1m)  {$n_1^+$};
\node[below right= of q1] (n1p)  {$n_1^-$};
\node[below right= of n1m] (q2)  {$m_1$};

\node[right= of q2] (qk)  {$m_{k-1}$};

\node[above right= of qk] (nkm)  {$n_k^+$};
\node[below right= of qk] (nkp)  {$n_k^-$};
\node[below right= of nkm] (qf)  {$m_k$};

\node[right= of qf] (nf)  {$n_\fin$};

\draw[->] (init) -- (q1);
\draw[->] (q1) to node [above] {1} (n1m);
\draw[->] (q1) to node [below] {0} (n1p);
\draw[->] (n1m) -- (q2);
\draw[->] (n1p) -- (q2);
\draw[dotted] (q2) -- (qk);
\draw[->] (qk) to node [above] {1} (nkm);
\draw[->] (qk) to node [below] {0} (nkp);
\draw[->] (nkm) -- (qf);
\draw[->] (nkp) -- (qf);
\draw[->] (qf) -- (nf);

\node[below=2 of E] (Vj) {$V_j:$};
\node[right= of Vj]  {$n_\init \longrightarrow t'_j\longrightarrow t_{j+1} \longrightarrow  n_\fin$};

\node[below= of Vj] (Tjd) {$T_{j,d:}$};
\node[right= of Tjd]  {$n_\init \longrightarrow t_j\longrightarrow m_{j,d} \longrightarrow  r_{j,d} \longrightarrow  n_\fin$};

\node[below= of Tjd] (Tjd') {$T'_{j,d:}$};
\node[right= of Tjd']  {$n_\init \longrightarrow n_{j,d} \longrightarrow  t'_j \longrightarrow  n_\fin$};

\node[below=2 of Tjd'.west,anchor=west] (Pjd) {$P_{j,d}$ (with $l_{j,d}=x_i$):};
\node[right= of Pjd] (Pjd0){$n_\init$};
\node[right= of Pjd0] (Pjd1) {$m_{i-1}$}; 
\node[above right= of Pjd1] (Pjd2) {$n_i^+$};
\node[below right= of Pjd1] (Pjd3) {$n_i^-$};
\node[right= of Pjd2] (Pjd4) {$r_{j,d}$};
\node[right= of Pjd3] (Pjd5) {$n_{j,d}$};
\node[right=5 of Pjd1] (Pjd_fin) {$n_\fin$};

\path (Pjd2) edge[->] (Pjd4) 
(Pjd3) edge[->] (Pjd5) 
(Pjd0) edge[->] (Pjd1)
(Pjd1) edge[->] (Pjd2)
(Pjd1) edge[->] (Pjd3)
 ($(Pjd4.east) +(1,-.5)$) edge[->] (Pjd5)
  ($(Pjd5.east) +(1,.5)$) edge[->] (Pjd4)
 ;
 
\draw[->]
(Pjd4.east) -- ++(1, -.5) -- (Pjd_fin);
\draw[->]
(Pjd5.east) -- ++(1, .5) -- (Pjd_fin);

\node[below=3 of Pjd.west,anchor=west] (P) {$P_{j,d}$ (with $l_{j,d}=\overline{x_i}$):};
\node[right= of P] (P0){$n_\init$};
\node[right= of P0] (P1) {$m_i$};
\node[above right= of P1] (P2) {$n_i^+$};
\node[below right= of P1] (P3) {$n_i^-$};
\node[right= of P2] (P4) {$r_{j,d}$};
\node[right= of P3] (P5) {$n_{j,d}$};
\node[right=5 of P1] (P_fin) {$n_\fin$};

\path (P2) edge[->] (P5) 
(P3) edge[->] (P4) 
(P0) edge[->] (P1)
(P1) edge[->] (P2)
(P1) edge[->] (P3)
 ($(P4.east) +(1,-.5)$) edge[->] (P5)
  ($(P5.east) +(1,.5)$) edge[->] (P4)
 ;
 
\draw[->]
(P4.east) -- ++(1, -.5) -- (P_fin);
\draw[->]
(P5.east) -- ++(1, .5) -- (P_fin);

\end{tikzpicture}
\caption{Graphs of the processes of $\Nn$.}
\label{fig:processes}
\end{figure}

Observe that in the graph for the non-deterministic process $P_{j,d}$ we assume
$\ell_{j,d} \in\set{x_i,\overline{x_i}}$. After $n_\init$, the process goes to $m_i$, and then to $n_i^+$ or $n_i^-$ in a deterministic way, depending on the result chosen at $m_i$. 
The rest of its behavior depends on whether $\ell_{j,d}=x_i$ or $\ell_{j,d}=\overline{x_i}$.
If $\ell_{j,d}=x_i$, then after $n^+_i$ the process goes to
$r_{j,d}$, and then to one of $\{n_\fin, n_{j,d}\}$ (nondeterminism!). For the other case, 
see Figure \ref{fig:processes}.

Process $P_{j,d}$ is designed with the following purpose. 
If process $E$ sets literal $\ell_{j,d}$ to true, then $P_{j,d}$
(together with $T_{j,d}$) guarantees that $m_{j,d}$ is
executed before $n_{j,d}$; if $E$ sets literal $\ell_{j,d}$ to false, then $m_{j,d}$ and $n_{j,d}$ 
can occur in any order. In other words, in every successful run containing $n_i^+$:  if
$\ell_{j,d}=x_i$ then node $m_{j,d}$ appears before $n_{j,d}$; if $\ell_{j,d}=\overline{x_i}$
then the nodes $m_{j,d}$ and $n_{j,d}$ can appear in any order. Similarly for runs containing $n_i^-$, 
interchanging $x_i,\overline{x_i}$. 







  

It is easy to see that $\Nn$ is very weakly non-deterministic.
The only nondeterministic processes are the $P_{j,d}$ processes. 
Moreover, the sets $\{n_\fin, r_{j,d}\}$ and $\{n_\fin, n_{j,d}\}$
are the only two sets of atomic negotiations such that (a) there are configurations
such that $P_{j,d}$ is ready to engage in them, and (b) contain more than one element.
Since $n_\fin$ contains all processes, the condition for very weak non-determinism is clearly verified.

If the valuation chosen by $E$ makes all $c_j$ true, we claim that the
partial run corresponding to this valuation cannot be completed to a
successful run. Indeed, in this case for each $c_j$ there is a true literal
$\ell_{j,d}$, and so $P_{j,d}$ enforces that $m_{j,d}$ is
executed before $n_{j,d}$. We denote this by $m_{j,d} < n_{j,d}$. 
But then, since process $T_{j,d}$ enforces $t_j < m_{j,d}$, process $T'_{j,d}$ enforces 
$n_{j,d} < t_j'$, and process $V_j$ enforces $t_j' < t_{j+1 \it mod(m+1)}$, we get the cycle
$$t_1 < t'_1 < t_2 < \cdots t_m < t'_m < t_1$$ 
Since, say, $t_2$ cannot occur before and after $t_1$, because the deterministic process are acyclic,
the partial run cannot be completed. 

Otherwise, if the valuation makes at least one $c_j$
false, then no cycle is created and the partial run can be extended to a successful run.


Here is a more formal version of the proof. By construction,
$\dom(t_j)=\set{V_{j-1},T_{j,d} \mid d=1,2,3}$, $\dom(t'_j)=\set{V_j,T'_{j,d}
  \mid d=1,2,3}$, $\dom(m_{j,d})=\set{T_{j,d}}$,
$\dom(n_{j,d})=\set{T'_{j,d},P_{j,d}}$,
$\dom(r_{j,d})=\set{T_{j,d},P_{j,d}}$.

Let $\n :\set{1,\ldots,k} \to \set{0,1}$ be a valuation of the
variables $x_1,\ldots,x_k$. By $C_\n$ we denote the following
configuration (note that only the position of processes $P_{j,d}$
depends on the valuation):
\begin{itemize}
\item $C_\n(E)=n_\fin$
\item If the literal $\ell_{j,d}$ is true under $\n$ then
  $C_\n(P_{j,d})=r_{j,d}$, and otherwise $C_\n(P_{j,d})=n_{j,d}$.
\item $C_\n(T_{j,d})=t_j$, $C_\n(T'_{j,d})=n_{j,d}$
\item $C_\n(V_j)=t'_j$
\end{itemize}

The following property is easy to check:

\smallskip

\noindent
\emph{Fact~1.} Every configuration $C_\n$ is reachable. Moreover,
every maximal run has some trace-equivalent prefix that reaches one of
the configurations $C_\n$. 


\smallskip

\noindent
\emph{Fact~2.} If $\n$ does not satisfy the formula then there
is a successful run from $C_\n$.

\smallskip

Assume that $\n$ does not satisfy clause $c_j$. By Fact~1, 
nodes $n_{j,d}$, $d=1,2,3$, are enabled. By executing these nodes we can
reach a configuration  with $P_{j,d}$ in $\set{r_{j,d},n_\fin}$ and
$T'_{j,d}$ in $t'_j$. Now $t'_j$ is executable, and $V_j$ moves to node
$t_{j+1}$. Notice that $t_{j+1}$ is now executable, so  $T_{j+1,d}$
can move first to $m_{j+1,d}$, then to $r_{j+1,d}$, $d=1,2,3$ (and
$V_j$ goes to
$n_\fin$). Now either $P_{j+1,d}$ is already in $r_{j+1,d}$ or it can
come there after $n_{j+1,d}$ is executed. In the first case nodes
$r_{j+1,d},n_{j+1,d}$ can be executed (in this order), in the second
case $n_{j+1,d},r_{j+1,d}$ can be executed (in this order). In either case we can get to a
configuration where processes $T_{j+1,d}$ are in $n_\fin$ and 
processes $T'_{j+1,d}$ are in $t'_{j+1}$. Therefore, $t'_{j+1}$ is
also executable. By iterating this
argument we obtain a successful run executing $t'_j,t_{j+1}, t'_{j+1}, \ldots,
t_1, t'_1,\ldots,t_j$ in this order. 

\smallskip

\noindent
\emph{Fact~3.} If $\n$ does satisfy the formula then there
is no successful run from $C_\n$.

\smallskip

Suppose by contradiction that there is a successful
run $\s$ from $C_\n$. By assumption, for each $j$ there is some
$d=1,2,3$ such that $C_\n(P_{j,d})=r_{j,d}$. Therefore $m_{j,d}$ is
necessarily executed before $n_{j,d}$ in $\s$. By construction this
implies that $t_j$ is executed before $t'_j$ in $\s$. 
Also by construction, $t'_j$ must be executed before
$t_{(j+1)\mathit{mod}(m+1)}$, for 
every $j$.
This means that $t_1$ should be executed before $t'_m$, and $t'_m$
before $t_1$, a contradiction.
\end{proof}


\section{Beyond soundness}
\label{sec:beyond-sound} 

We show that the techniques we have developed for the soundness problem, 
like forks and the omitting game, can be used to 
check other functional properties of acyclic deterministic negotiations. 
Moreover, 
we show that soundness reduces the complexity of verifying these properties: 
while checking the property for arbitrary deterministic negotiations---sound or not---is an 
intractable problem, it is polynomial in the sound case. 

In Section \ref{subsub:races} we study the race problem: Given two atomic negotiations
$m$ and $n$, is there a run in which they occur concurrently? In Section \ref{subsub:data}
we address the static analysis of negotiations with data. We assume that the result 
of an atomic negotiation corresponds to some operations acting on a
set of variables.
We study some standard
questions of static analysis, for instance whether a variable can be allocated and then never
deallocated before the execution ends. These questions can be answered in exponential
time by constructing the reachability graph of the negotiation and applying standard model
checking algorithms. (This is the approach followed in~\cite{aalst}, which studies the questions
for workflow Petri nets.) 
We exhibit polynomial algorithms for the acyclic deterministic case.

\begin{defi}
\label{def:race}
Let $\Nn=\struct{\Proc,N,\dom,R,\d}$ be a negotiation. Two atomic negotiations
$m, n \in N$ can be \emph{concurrently enabled}, denoted $m\para n$, if 
$\dom(m) \cap \dom(n) = \emptyset$ and there is a reachable configuration $C$
of $\Nn$ where both $m$ and $n$ are enabled.
\end{defi}

\subsection{Races}
\label{subsub:races}

For a given pair of atomic negotiations $m,n$ of a deterministic negotiation
$\Nn=\struct{\Proc,N,\dom,R,\d}$, we want to determine if 
there is a \emph{race} between $m$ and $n$, i.e., if there is a reachable configuration
that concurrently enables $m$ and $n$.

This question was answered in \cite{KovEsp} for live and safe free-choice nets,
where a polynomial fixed point algorithm was given. The algorithm can also be applied
to sound deterministic negotiations (cyclic or acyclic), but has cubic complexity 
in the number of atomic negotiations. We show that in the acyclic case there is a simple
anti-pattern characterization of the race pairs $m, n$, which leads to
an algorithm
that runs in linear time and logarithmic space.

In the rest of the section we give a syntactic characterization of the pairs $m, n$ such that $m \para n$. 
We proceed in several steps. 
Lemma \ref{p:acyclic-parallel} shows the equivalence of 
the semantic definition with a more concrete, but still semantic, condition. 
Proposition \ref{p:reach} transforms 
this condition into the conjunction of a syntactic and a semantic condition.
Proposition \ref{p:localreach} replaces the latter by the existence of a certain fork.
The final result, given in Theorem \ref{thm:races}, just puts the two propositions together.

Recall that two runs $w,w' \in (N \times R)^*$ are equivalent if 
$w'$ can be obtained from $w$ by repeatedly exchanging adjacent pairs
$(m,a)(n,b)$ into $(n,b)(m,a)$ whenever $\dom(m) \cap \dom(n)=\es$.

\begin{lem}\label{p:acyclic-parallel}
 Let $\Nn$ be an acyclic, deterministic, sound negotiation, and let
  $m,n$ be two atomic negotiations in $\Nn$. Then $m \para n$ iff
  \emph{every} run $w$ from $n_\init$ containing both $m$ and $n$ has
  an equivalent run $w'=w_1w_2$ such that $w' = C_\init \act{w_1} C \act{w_2}
  C'$ for some configuration $C$ where both $m$ and $n$ are enabled.
\end{lem}
{
\begin{proof}
It suffices to show the implication from left to right. So assume
  that there exists some reachable configuration $C_0$ where both $m$
  and $n$ are enabled.  In particular, $\dom(m) \cap \dom(n)=\es$. 
By way of contradiction, let us suppose that there
  exists some run containing both $m$ and $n$, but this run cannot be
  reordered as in the statement of the lemma. 

Assuming that $m$ appears before $n$ in this run,
we claim that there must be some local path 
from $m$ to $n$ in $\Nn$. To see this, assume the contrary and
consider a run of the form $w=w_1 (m,a)w_2 (n,b) w_3$. The run $w$
defines a partial order (actually a Mazurkiewicz trace) $\tr(w)$ with nodes corresponding
to positions in $w$, and edges from $(m',c)$ to $(n',d)$ if
$\dom(m')\cap \dom(n')\not=\es$ and $(m',c)$ precedes $(n',d)$ in
$w$. Since there is no path from $m$ to $n$ in $\Nn$, nodes $(m,a)$
and $(n,b)$ are unordered in $\tr(w)$. So we can choose a topological order
$w'$ of $\tr(w)$ of the form $w'=w'_1 (m,a) (n,b) w'_2$. This shows the claim.

So let $\pi$ be a path in $\Nn$ from $m$ to $n$, say $m,
n_1,\ldots,n_k,n$. Let $p$ be 
some process such that $n_k \act{p,a'} n$ for some result $a'$.
  Let us go back to $C_0$. Since both $m$ and $n$ are enabled in $C_0$, we
  have a transition $C_0 \act{n,b} C_1$, for some $b \in\out(n)$. Note that $m$ is still enabled in $C_1$, since
  $\dom(m)\cap\dom(n)=\es$. So we can apply Lemma~\ref{l:local} to
  $C_1$ and $\pi$ (because $\Nn$ is sound), obtaining a
  configuration $C_2$ where $C_2(p)=n$. But since $n$ was
  executed before $C_1$, this violates the acyclicity of $\Nn$.
\end{proof}
}

The next step is to convert the condition from
Lemma~\ref{p:acyclic-parallel} to a condition on the graph of a
negotiation. 
This condition uses the notion of fork from Definition~\ref{def:fork}.
\begin{prop}\label{p:reach}
 Let $\Nn$ be an acyclic, deterministic, sound negotiation, and let
  $m,n$ be two atomic negotiations in $\Nn$ with $\dom(m) \cap \dom(n)=\es$. Then $m \para n$ if{}f
  there exists an initial  run containing both $m,n$, and 
  there is neither a local path from $m$ to $n$ nor a local path from $n$ to $m$.
  \end{prop}
\SKIP{
\begin{proof}
For the left-to-right implication, assume by contradiction that there
is some local path $\pi$ from $m$ to $n$. Consider some reachable configuration
$C$ such that $m$ is enabled in $C$. By Lemma~\ref{l:local} we also find a run $C
\act{*} C'$ such that $n$ is enabled in $C'$. But note that the run $C_\init \act{*} C
\act{*} C'$ cannot be reordered as stated in
Lemma~\ref{p:acyclic-parallel}, a contradiction.

For the converse, consider some run $w$ containing both $m,n$. Since
there are no local paths in $\Nn$ between $m,n$, we can reorder, as in
the proof of Lemma~\ref{p:acyclic-parallel}, the run $w$
into some $w'$
such that we find a configuration $C$ of $w'$ where both $m$ and $n$
are enabled. 
\end{proof}
}

\begin{prop}\label{p:localreach}
Let $\Nn$ be a sound deterministic negotiation, and let $m,n$ be two
atomic negotiations of $\Nn$. 
Then $\Nn$ has an initial run containing both $m$ and $n$ if{}f it has a fork $(p, q, m, n)$ for some $p,q$.
\end{prop}
  \begin{proof}
    Right-to-left implication: the proof is similar to the one of
    Lemma~\ref{l:local}, but we need to consider three paths instead
    of a single one. First we realize the
    path from $n_\init$ to the branching node $n'$ (with result $a$) in the definition of fork, using Lemma~\ref{l:local}. Suppose that the run from $C_\init$
    to the configuration $C$ that enables $n'$, contains
    neither $m$ nor $n$ (otherwise another application of
    Lemma~\ref{l:local} suffices). Let $C \act{(n',a)} C_1$. We show
    by induction on the sum of the lengths of the two local paths of the fork
    how to construct a run containing both $m$ and $n$.  Let $C_1(p)=m_0$, $C_1(q)=n_0$ and let
\begin{equation*}
  m_0\act{p,a_0}\cdots\act{p,a_{k-1}} m_k=m
  \qquad\text{and}\qquad
  n_0\act{q,b_0}\cdots\act{q,b_{l-1}} n_l=n\,.
\end{equation*}
be the local paths of the fork. Since $\Nn$ is sound, some run leads from
$C_1$ to the final configuration. Since $\Nn$ is $p$ in $C_1$ can do
only $m_0$, and $q$ only $n_0$, the run  must
execute both $m_0$ and $n_0$ at some point. Suppose without loss of generality
 that $m_0$ is executed first. We have  $C_1 \act{*} C_2 \act{(m_0,a_0)} C_3$ for some $C_2,C_3$ such
that $C_3(p)=m_1$ (since $\Nn$ is deterministic) and $C_3(q)=n_0$. We can now apply the induction
hypothesis to the local paths $m_1 \act{*} m$, $n_0 \act{*} n$, and we
are done.

Left-to-right implication: By assumption $\Nn$ has a run from $n_\init$ of the form
$w=w_1 (m,b) w_2 (n,c)$. Choose some $p \in \dom(m), q
\in\dom(n)$. Let $(n',a)$ be the rightmost letter of $w_1$ such that
$\set{p,q} \subseteq \dom(n')$ (which exists because $\set{p,q} \subseteq \dom(n_\init)$ by definition), 
and let $m_0=\d(n',a,p)$, $n_0=\d(n',a,q)$. Then we can extract from $w_1$ a $p$-path leading from $m_0$ to
$m$, and from $w_1w_2$ a $q$-path leading from $n_0$ to $n$.  
By the choice of $n'$ these two paths are disjoint.
\end{proof}

From Proposition~\ref{p:reach} and \ref{p:localreach}  we immediately obtain:

\begin{thm}
\label{thm:races}
  For any acyclic, deterministic, sound negotiation $\Nn$ we can
  decide in linear time (resp., in logarithmic space)  whether two atomic negotiations $m,n$ of $\Nn$
  satisfy $m \para n$. The above problem is \NLOGSPACE-complete. 
\end{thm}

It is not difficult to show that soundness is essential for this result. Indeed, the race problem 
is \NP-hard for acyclic and deterministic, but not necessarily sound, negotiations. We sketch a reduction 
from 3CNF-SAT. Given a 3CNF-formula $\phi$, we construct a deterministic negotiation $\Nn_\phi$ 
with distinguished atomic negotiations $n$, $m$ and $n_\mathit{true}$
satisfying two properties: (a) $\Nn_\phi$ has a run enabling
$n_\mathit{true}$ if{}f $\phi$ is satisfiable, and (b) $n$ and $m$ are
concurrently enable iff $n_\mathit{true}$ is executed.
Clearly, (a) and (b) imply that there is a race between $n$ and $m$ if{}f $\phi$ is satisfiable.

$\Nn_\phi$ has a process  $C_j$ for each clause $c_j$, and a proces $P_{i,j}$ for each clause $c_j$ and variable $x_i$ occurring in $c_j$ (positively or negatively). Besides the initial and final atomic negotiations,
the nodes of $\Nn_\phi$ are:
\begin{itemize}
\item a node ${\it SetVariable}_i$ for each variable $x_i$, with all $P_{i, j}$ as participants and two outcomes $\texttt{true}$ and $\texttt{false}$; 
\item a node ${\it GuessLiteral}_j$ for each clause $c_j$, with $C_j$ as only participant and three outcomes, one for each literal of $c_j$; 
\item a node ${\it TrueLiteral_{i, j}}$ for each clause $c_j$ and variable $x_i$ occurring in $c_j$, with $P_{i,j}$ and $C_j$ as participants and one outcome $\texttt{true}$; 
\item a node $n_\mathit{true}$, with all $C_j$ as participants, and also one outcome $\texttt{true}$; and 
\item nodes $n$ and $m$, with participants $C_1$ and $C_2$, respectively. 
\end{itemize}
The transition function is designed so that $\Nn_\phi$ behaves as
follows. The initial atomic negotiation concurrently enables all nodes
${\it SetVariable}_i$ and ${\it GuessLiteral}_j$. At
${\it SetVariable}_i$ processes $P_{i, j}$'s jointly select a truth
assignment for $x_i$, and process $P_{i,j}$ gets ready to engage in 
${\it TrueLiteral_{i, j}}$ if the chosen assignment makes the
corresponding literal true. 
At ${\it GuessLiteral}_j$ process $C_j$ guesses a literal of $c_j$,
say with variable $i$, and gets ready to engage in ${\it TrueLiteral_{i,
j}}$.
If ${\it TrueLiteral_{i, j}}$ gets enabled, action $\texttt{true}$ is
executed and $C_j$ becomes ready to engage in $n_\mathit{true}$. 
So $n_\mathit{true}$ 
becomes enabled if{}f the assignment guessed at ${\it SetVariable}$
nodes  satisfies all clauses and at ${\it GuessLiteral}_j$ true
literals are guessed correctly. 
After $n_\mathit{true}$ occurs, processes $C_1$
and $C_2$ become ready to engage in $n$ and $m$, respectively, and so
$n$ and $m$ become enabled. Therefore, if $\phi$ is satisfiable, then
$n$ and $m$ can be concurrently enabled by selecting a satisfying
assignment and correctly guessing true literals. If $\phi$ is
unsatisfiable, then $n_\mathit{true}$ can never occur, and so $n$ and
$m$ cannot occur either. It is easy to see that the transition
function is deterministic.

\subsection{Negotiations with data}
\label{subsub:data}

A negotiation with data is a negotiation over a given
set $X$ of variables  (over finite domains).
Each result $(n,a) \in N \times R$
comes with a set $\S$ of operations on the shared variables. In our examples this set $\S$ is composed of
$\all(x)$, $\rd(x)$, $\wrt(x)$, and $\deall(x)$. 

Formally, a~\emph{negotiation with data} is a negotiation with one
additional component: 
$\Nn=\struct{\Proc,N,\dom,R,\d,\ell}$ where $\ell \colon (N \times R) \to \Pp({\S\times X})$ maps
every result to a (possibly empty) set of data operations on variables from $X$. We assume that 
for each $(n,a) \in N \times R$
and for each variable $x \in X$ the label $\ell(n,a)$ contains at most one operation 
on $x$, that is, at most one element of $\S \times \{x\}$.

As an example, we enrich the deterministic negotiation on the left of Figure~\ref{fig:example1} with data
(this example is adapted from ~\cite{DBLP:conf/caise/TrckaAS09}).
The negotiation is shown again in Figure~\ref{fig:example2}. Its results
perform operations on two variables $x_1$ and $x_2$. 
The table on the right of the figure gives for each result and for each operation
the (indices of the) variables to which the result applies this operation. Observe that
the atomic negotiation $n_5$ has one result that is not represented graphically.

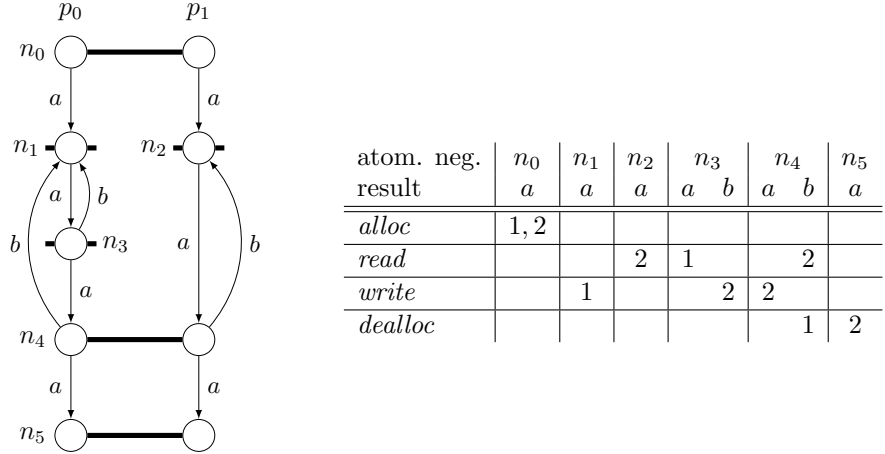
\begin{figure}[h]
\raisebox{-2.7cm}{\scalebox{0.85}{\begin{tikzpicture}
\nego[ports=2,id=n0,spacing=2]{0,0}
\node[left = 0cm of n0_P0, font=\large] {$n_0$};
\nego[ports=1,id=n1, spacing=2]{0,-1.5}
\node[left = 0.1cm of n1_P0, font=\large] {$n_1$};
\nego[ports=1,id=n2]{2,-1.5}
\node[left = 0.1cm of n2_P0, font=\large] {$n_2$};
\nego[ports=1,id=n3]{0,-3}
\node[right = 0.1cm of n3_P0, font=\large] {$n_3$};
\nego[ports=2,id=n4, spacing=2]{0,-4.5}
\node[left = 0cm of n4_P0, font=\large] {$n_4$};
\nego[ports=2,id=n5, spacing=2]{0,-6}
\node[left = 0cm of n5_P0, font=\large] {$n_5$};

\node[above = 0.1cm of n0_P0, font=\large] {$p_0$};
\node[above = 0.1cm of n0_P1, font=\large] {$p_1$};

\pgfsetarrowsend{latex}
\draw (n0_P0) edge node[left=0] {$a$} (n1_P0);
\draw (n0_P1) edge node[right=0] {$a$} (n2_P0);
\draw (n1_P0) edge node[left=0] {$a$} (n3_P0);
\draw (n2_P0) edge node[left=0] {$a$} (n4_P1);
\draw (n3_P0) edge node[right=0] {$a$} (n4_P0);
\draw (n3_P0) edge[bend right=30] node[right=0] {$b$} (n1_P0);
\draw (n4_P0) edge[bend left=40] node[left=0] {$b$} (n1_P0);
\draw (n4_P1) edge[bend right=40] node[right=0] {$b$} (n2_P0);
\draw (n4_P0) edge node[left=0] {$a$} (n5_P0);
\draw (n4_P1) edge node[right=0] {$a$} (n5_P1);

\end{tikzpicture}}} 
\qquad
{\small
\begin{tabular}{l|c|c|c|cc|cc|cccc}
atom. neg.    & $n_0$   & $n_1$ & $n_2$ & \multicolumn{2}{c|}{$n_3$} & \multicolumn{2}{c|}{$n_4$}  &  $n_5$ \\ 
result       & $a$     & $a$   & $a$   & $a$ & $b$                 & $a$ & $b$                  &  $a$   \\ \hline \hline
{\it alloc}   & $1,2 $  &       &       &     &                     &                            &        \\  \hline
{\it read}    &         &       & $2$   & $1$ &                     &     & $2$                  &        \\  \hline
{\it write}   &         & $1$   &       &     & $2$                 & $2$ &                               \\  \hline
{\it dealloc} &         &       &       &     &                     &     & $1$                  &  $2$   
\end{tabular}}
\caption{A negotiation operating on data.}
\label{fig:example2}
\end{figure}

In \cite{DBLP:conf/caise/TrckaAS09} some examples of data
specifications for workflows are considered. 

\begin{enumerate}
\item[(1)] {\it Inconsistent data}: an atomic negotiation reads or writes a variable $x$
while another atomic negotiation is writing, allocating, or
deallocating it in parallel. In our example there is a run
in which $(n_2, a)$ and $(n_3,b)$ read and write to $x_2$ in parallel. 
\item[(2)] {\it Weakly redundant data}: there is a run in which a variable is written and never read  
before it is deallocated or the run ends. In the example, there is a run in which $x_2$ is written by
$(n_3, b)$, and never read again before it is deallocated by $(n_5, a)$. 
\item[(3)] {\it Never destroyed}: there is an execution in which a variable is allocated and then never 
deallocated before the execution ends. The example has a run which never takes $(n_4,b)$. In this run the variable
$x_1$ is never deallocated.
\end{enumerate}

It is easy to give algorithms for these properties that are polynomial {\em in the size of the 
reachability graph}. We give the first algorithms that check these properties in 
polynomial time {\em in the size of the negotiation}, which can be exponentially smaller than 
its reachability graph. 

For the first property
we can directly use the algorithm for the race problem: It suffices to check 
if the negotiation has two
results $(m,a), (n,b)$ such that $m \para n$  and 
there is a variable $x$ such that $\ell(a) \cap \{\rd(x),\wrt(x) \} \neq \emptyset$ and
$\ell(b) \cap \{\wrt(x), \all(x), \deall(x) \} \neq \emptyset$.

In the rest of the section we present a polynomial algorithm for the
following abstract problem, which has 
the problems (2) and (3) above as special instances.

Given sets $\Oo_1, \Oo_2, \Oo \subseteq N \times R$ 
we say
that the negotiation $\Nn$ \emph{violates the specification $(\Oo_1,
  \Oo_2,\Oo)$} if there is a successful run $w = (n_1, a_1) \cdots (n_k, a_k)$  
of $\Nn$, and indices 
$0 \leq i < j \leq k$ such that $(n_i, a_i) \in \Oo_1$, $(n_j, a_j) \in
\Oo_2$, and $(n_l,a_l) \notin \Oo$ for all $i<l<j$. In this case we
also say that $(\Oo_1,
  \Oo_2,\Oo)$ is violated at $(n_i,a_i),(n_j,a_j)$. Otherwise
$\Nn$ \emph{complies with} $(\Oo_1, \Oo_2,\Oo)$.

\begin{exa}
In order to simplify notations below, we assume that $n_\fin$ has some
result - if not, we can add a self-loop to $n_\fin$. 

Observe that a variable $x$ is weakly redundant (specification of type
  (2)) if{}f $\Nn$ violates $(\Oo_1, \Oo_2,\Oo)$, 
where $\Oo_1= \set{(n,a) \in N \times R : \wrt(x) \in \ell(n,a)}$, 
$\Oo_2= \set{(n,b) \in N \times R : n = n_{{\it fin}} \vee \deall(x)
  \in \ell(n,b)}$ and $\Oo=\set{(n,c) : \ell(n,c) \cap
  (\S\times\set{x}) \not=\es}$.

A variable $x$ is never destroyed (specification of type (3)) iff $\Nn$ violates $(\Oo_1,
\Oo_2,\Oo)$, where $\Oo_1= \set{(n,a) \in N \times R : \all(x) \in
  \ell(n,a)}$, $\Oo_2= \set{n_{{\it fin}}}$, $\Oo=\set{(n,c) : 
   n=n_{{\it fin}} \vee \ell(n,c) \cap \set{\all(x),\deall(x)} \not=\es }$.
\end{exa}

For the next proposition it is convenient to use the notation $(m,a)
\act{+} (n,b)$, for $m,n \in N$, $a \in \out(m)$, $b \in\out(n)$ whenever there is a
(non-empty) local path in $\Nn$ from $\d(m,a,p)$ to $n$, for some $p
\in\dom(m)$. 

\begin{prop}\label{p:counterex}
  Let $\Nn$ be an acyclic, deterministic, sound negotiation with data,
  and $(\Oo_1, \Oo_2,\Oo)$ a specification. Let $(m,a) \in \Oo_1$,
  $(n,b) \in\Oo_2$.  Then $\Nn$ violates $(\Oo_1,\Oo_2,\Oo)$ at
  $(m,a),(n,b)$ iff $(n,b) \not\act{+} (m,a)$ and  
 $\Nn$
  has a successful run containing $P=\set{(m,a),(n,b)}$, and
  omitting the set $B=\set{(m',c) \in \Oo : (m,a) \act{+} (m',c)
    \act{+} (n,b)}$.
\end{prop}
\SKIP{
\begin{proof}
  For the right-to-left direction: assume that $\Nn$ has a run $w$ as
  claimed. Since $(n,b) \not\act{+} (m,a)$,  $w$ can be assumed to be of the form $w=w_1
  (m,a) w_2 (n,b) w_3$. By reordering $w$ we may suppose that
  for every $(m',c)$ in $w_2$, we have $(m,a) \act{+} (m',c) \act{+} (n,b)$. Thus,
  since $w$ omits $B$ this means that $(m',c) \notin \Oo$, so the
  claim follows.

For the left-to-right direction: if $\Nn$ violates $(\Oo_1,\Oo_2,\Oo)$ at
  $(m,a),(n,b)$ then there is a run $w = (n_1, a_1) \cdots (n_k,
  a_k)$ with $(n_i,a_i)=(m,a)$, $(n_j,a_j)=(n,b)$ and such that
  $(n_l,a_l) \notin \Oo$ for all $i<l<j$. Since $(\set{(m',c) : (m,a)
    \act{+} (m',c) \act{+} (n,b)} \cap \set{n_i : 1\leq i\leq k}) \subseteq
  \set{n_l : i<l<j}$, the run $w$ contains $(m,a),(n,b)$ and omits
  $B$. 
\end{proof}
}

\begin{rem}
  Note that Proposition~\ref{p:counterex} refers to omitting pairs of
atomic negotiation/result, whereas the original omitting problem
refers to omitting atomic negotiations. However, it is straightforward to adapt
the omitting problem and the corresponding results as to handle pairs.
\end{rem}

Putting together Proposition~\ref{p:counterex} and Theorem~\ref{thm:omitting}
we obtain:

\begin{cor}
  Given an acyclic, deterministic, sound negotiation with data  $\Nn$,
  and  a specification $(\Oo_1, \Oo_2,\Oo)$, it can be checked in
  polynomial time whether $\Nn$ complies with $(\Oo_1, \Oo_2,\Oo)$.
\end{cor}

\section{Conclusions}\label{sec:conc}

Verification questions for finite-state of concurrent systems are very often \PSPACE-hard because
of the state explosion problem. One approach to this challenge is to search for restrictions in the
communication primitives that permit non-trivial interactions, and yet are 
algorithmically easier to analyze. We have shown that negotiations 
are a suitable model for this line of attack. On the one hand, non-deterministic processes can simulate 
any other communication primitive (up to reasonable equivalences); on the other hand,
the definition of non-determinism immediately suggests to investigate the deterministic and weakly 
non-deterministic classes. Even the deterministic negotiations have enough expressive power for interesting
applications in the workflow modeling domain. Indeed, they are equivalent to workflow free-choice nets 
\cite{DBLP:journals/topnoc/DeselE16}, and acyclicity and 
free-choiceness are quite common: about 70\% of the industrial workflow nets 
of ~\cite{van2007verification,fahland2009instantaneous,DBLP:conf/fase/EsparzaH16} 
are free-choice, and about 60\% are both acyclic and free-choice (see e.g. the table at 
the end of \cite{DBLP:conf/fase/EsparzaH16}). We have shown that a number of verification 
problems for sound deterministic negotiations can be solved in \PTIME\ or even in
\NLOGSPACE. 

\subsection*{Connection to workflow Petri nets.} The connection between negotiations and Petri nets is
studied in detail in \cite{DBLP:journals/topnoc/DeselE16}. We explain the consequences of our results for 
workflow nets. The following discussion assumes that the reader is familiar with workflow Petri nets and the 
concept of an S-component of a net.

Let us explain how using the results of this paper
we can give a \NLOGSPACE-algorithm to check soundness of free-choice workflow nets, assuming that the input
is not only the workflow net itself, but also a certain decomposition of the net into S-components.

Before proceeding, we need to examine the definitions of soundness used in this paper and 
in \cite{DBLP:journals/topnoc/DeselE16}. There exist many notions of soundness for workflow Petri nets 
\cite{DBLP:journals/fac/AalstHHSVVW11}. Soundness as defined in this paper corresponds to weak soundness 
in \cite{DBLP:journals/fac/AalstHHSVVW11}. Fortunately, it has been shown in \cite{DBLP:journals/tsc/Liu14a} that 
soundness and weak soundness coincide for free-choice workflow Petri nets, and so we do not need to worry about the
differences between the two.

The following discussion refers to results
in~\cite{DBLP:journals/topnoc/DeselE16}. 
There, Desel and Esparza describe a mapping that assigns to every 
negotiation a so-called in/out net (Proposition 8) satisfying the following properties:
\begin{enumerate}
\item the in/out net is \emph{decomposable}: it admits a covering by 
S-components such that any two distinct S-components
share the initial and final places, and no other place. 
\item the in/out net is only linearly larger than the deterministic negotiation,
and can be constructed in linear time. 
\end{enumerate}
\noindent Further, \cite{DBLP:journals/topnoc/DeselE16} proves the following two results:
\begin{enumerate}[label=(\alph*)]
\item The in/out net of a sound deterministic negotiation is a sound free-choice workflow net
(Corollary 4 and Proposition 7).
\item Every sound decomposable free-choice workflow net is the in/out net of a sound deterministic negotiation
(Proposition 9). 
\end{enumerate}
\noindent Finally, it follows easily from the proof of Proposition 9 that the sound deterministic negotiation
of (b) can be constructed in linear time and logarithmic space from the net \emph{and} its covering by S-components.
This shows that the following problem is in \NLOGSPACE: Given a decomposable free-choice workflow net $W$, and a decomposition of $W$ into S-components, is $W$ sound?

Our results on acyclic negotiations also have consequences for workflow nets.
It is known that soundness and weak soundness of 1-safe acyclic workflow Petri nets 
are \coNP-complete, and so most likely not solvable in polynomial time \cite{DBLP:journals/tsmc/TipleaBC15}. 
Since the in/out nets of acyclic negotiations are also acyclic, 
Theorem \ref{th:weak-sound} of the present paper identifies a proper superclass of acyclic free-choice workflow nets 
for which weak soundness remains polynomial: the workflow nets derived from acyclic, weakly 
nondeterministic negotiations. This class does not coincide with any of the known Petri net classes beyond 
free-choice, like asymmetric choice nets. We see this as a point in favor of the negotiation model:
it can be used to investigate classes of nets beyond the free-choice class for which some important 
properties can be efficiently decided.

\subsection*{Future work.}  There are several open problems we are currently working on, or intend to address.
We do not know if the soundness problem for acyclic weakly 
non-deterministic negotiations is \PTIME-complete.  Also, we conjecture that the polynomial algorithms
for acyclic deterministic negotiations of Section \ref{sec:beyond-sound} can be extended to the
cyclic case. More generally, we would like to have a better understanding of which 
verification problems for sound deterministic negotiations can be solved in \PTIME.  
Since these negotiations are not closed under products with automata,
we should not expect to be able to polynomially decide arbitrary safety properties. 
Finally, the analogous question for sound weakly non-deterministic negotiations and the class
\NP\ is also increasingly interesting, due to the important advances in SMT tools.

\subsection*{Acknowledgements} We thank the referees for careful
reading of the paper. 
\igwin{added Acknowledgements}



\bibliographystyle{plainurl}
\bibliography{biblio}

\begin{thebibliography}{10}

\bibitem{DBLP:journals/topnoc/DeselE16}
J{\"{o}}rg Desel and Javier Esparza.
\newblock Negotiations and petri nets.
\newblock {\em T. Petri Nets and Other Models of Concurrency}, 11:203--225,
  2016.

\bibitem{DR95}
Volker Diekert and Grzegorz Rozenberg, editors.
\newblock {\em {The Book of Traces}}.
\newblock World Scientific, Singapore, 1995.

\bibitem{negI}
Javier Esparza and J{\"o}rg Desel.
\newblock On negotiation as concurrency primitive.
\newblock In {\em CONCUR}, pages 440--454, 2013.
\newblock Extended version in CoRR abs/1307.2145.

\bibitem{negII}
Javier Esparza and J{\"o}rg Desel.
\newblock On negotiation as concurrency primitive {I}{I}: Deterministic cyclic
  negotiations.
\newblock In {\em FoSSaCS}, pages 258--273, 2014.
\newblock Extended version in CoRR abs/1403.4958.

\bibitem{DBLP:conf/fase/EsparzaH16}
Javier Esparza and Philipp Hoffmann.
\newblock Reduction rules for colored workflow nets.
\newblock In {\em FASE}, volume 9633 of {\em LNCS}, pages 342--358. Springer,
  2016.

\bibitem{ekmw16}
Javier Esparza, Denis Kuperberg, Anca Muscholl, and Igor Walukiewicz.
\newblock Soundness in negotiations.
\newblock In {\em 27th International Conference on Concurrency Theory, {CONCUR}
  2016, August 23-26, 2016, Qu{\'{e}}bec City, Canada}, volume~59 of {\em
  LIPIcs}, pages 12:1--12:13. Schloss Dagstuhl - Leibniz-Zentrum f{\"u}r
  Informatik, 2016.

\bibitem{fahland2009instantaneous}
Dirk Fahland, C{\'e}dric Favre, Barbara Jobstmann, Jana Koehler, Niels Lohmann,
  Hagen V{\"o}lzer, and Karsten Wolf.
\newblock Instantaneous soundness checking of industrial business process
  models.
\newblock In {\em Business Process Management}, pages 278--293. Springer, 2009.

\bibitem{DBLP:conf/tacas/FavreVM16}
C{\'{e}}dric Favre, Hagen V{\"{o}}lzer, and Peter M{\"{u}}ller.
\newblock Diagnostic information for control-flow analysis of workflow graphs
  (a.k.a.~free-choice workflow nets).
\newblock In {\em TACAS 2016}, volume 9636 of {\em LNCS}, pages 463--479.
  Springer, 2016.

\bibitem{GKM06}
Blaise Genest, Dietrich Kuske, and Anca Muscholl.
\newblock A {K}leene theorem and model checking algorithms for existentially
  bounded communicating automata.
\newblock {\em Inf.\& Comput.}, 204(6):920--956, 2006.

\bibitem{KovEsp}
Andrei Kovalyov and Javier Esparza.
\newblock A polynomial algorithm to compute the concurrency relation of
  free-choice signal transition graphs.
\newblock In {\em Workshop on Discrete Event Systems, WODES'96}, pages 1--7.
  Institute of Electrical Engineers, 1996.

\bibitem{DBLP:journals/tsc/Liu14a}
GuanJun Liu.
\newblock Some complexity results for the soundness problem of workflow nets.
\newblock {\em {IEEE} Trans. Services Computing}, 7(2):322--328, 2014.

\bibitem{mus15}
Anca Muscholl.
\newblock Automated synthesis of distributed controllers.
\newblock In {\em ICALP 2015}, volume 9135 of {\em LNCS}, pages 11--27.
  Springer, 2015.

\bibitem{DBLP:journals/is/SidorovaST11}
Natalia Sidorova, Christian Stahl, and Nikola Trcka.
\newblock Soundness verification for conceptual workflow nets with data: Early
  detection of errors with the most precision possible.
\newblock {\em Inf. Syst.}, 36(7):1026--1043, 2011.

\bibitem{Thomas08}
Wolfgang Thomas.
\newblock Church's problem and a tour through automata theory.
\newblock In {\em Pillars of Computer Science, Essays Dedicated to Boris (Boaz)
  Trakhtenbrot on the Occasion of His 85th Birthday}, volume 4800 of {\em
  LNCS}, pages 635--655. Springer, 2008.

\bibitem{DBLP:journals/tsmc/TipleaBC15}
Ferucio~Laurentiu Tiplea, Corina Bocaneala, and Raluca Chirosca.
\newblock On the complexity of deciding soundness of acyclic workflow nets.
\newblock {\em {IEEE} Trans. Systems, Man, and Cybernetics: Systems},
  45(9):1292--1298, 2015.

\bibitem{DBLP:conf/caise/TrckaAS09}
Nikola Trcka, Wil M.~P. van~der Aalst, and Natalia Sidorova.
\newblock Data-flow anti-patterns: Discovering data-flow errors in workflows.
\newblock In {\em Advanced Information Systems Engineering (CAiSE)}, volume
  5565 of {\em LNCS}, pages 425--439. Springer, 2009.

\bibitem{aalst}
W.~M.~P. van~der Aalst.
\newblock The application of {P}etri nets to workflow management.
\newblock {\em J. Circuits, Syst. and Comput.}, 08(01):21--66, 1998.

\bibitem{DBLP:conf/bpm/Aalst00}
Wil M.~P. van~der Aalst.
\newblock Workflow verification: Finding control-flow errors using
  petri-net-based techniques.
\newblock In Wil M.~P. van~der Aalst, J{\"{o}}rg Desel, and Andreas Oberweis,
  editors, {\em Business Process Management, Models, Techniques, and Empirical
  Studies}, volume 1806 of {\em Lecture Notes in Computer Science}, pages
  161--183. Springer, 2000.

\bibitem{DBLP:journals/fac/AalstHHSVVW11}
Wil M.~P. van~der Aalst, Kees~M. van Hee, Arthur H.~M. ter Hofstede, Natalia
  Sidorova, H.~M.~W. Verbeek, Marc Voorhoeve, and Moe~Thandar Wynn.
\newblock Soundness of workflow nets: classi\-fication, decidability, and
  analysis.
\newblock {\em Formal Asp. Comput.}, 23(3):333--363, 2011.

\bibitem{van2004workflow}
Wil M.~P. van~der Aalst and Kees~Max van Hee.
\newblock {\em Workflow management: models, methods, and systems}.
\newblock MIT press, 2004.

\bibitem{van2007verification}
B.~{van Dongen}, M.~Jansen-Vullers, {HMW}. Verbeek, and Wil~MP. {van der
  Aalst}.
\newblock Verification of the {SAP} reference models using {EPC} reduction,
  state-space analysis, and invariants.
\newblock {\em Computers in Industry}, 58(6):578--601, 2007.

\bibitem{zie87}
W.~Zielonka.
\newblock Notes on finite asynchronous automata.
\newblock {\em RAIRO--Theoretical Informatics and Applications}, 21:99--135,
  1987.

\end{thebibliography}

\end{document}